\documentclass[preprint, 3p, authoryear]{elsarticle}

\usepackage{hyperref}
\usepackage{amsmath}
\usepackage{amsthm}
\usepackage{amssymb}
\usepackage{cancel}
\usepackage[authoryear]{natbib}
\usepackage{dsfont}
\usepackage{amssymb}
\usepackage{amsfonts}
\usepackage{amsthm}
\usepackage{epsfig}
\usepackage{setspace}
\usepackage{comment}
\usepackage{framed}
\usepackage{graphicx}
\usepackage{multirow,algorithm2e}
\usepackage{bbm}

\usepackage{lineno}
\modulolinenumbers[5]
\usepackage{graphicx} 

\usepackage{comment}
\usepackage[dvipsnames]{xcolor}

\makeatletter
\newcommand{\AlgoResetCount}{\renewcommand{\@ResetCounterIfNeeded}{\setcounter{AlgoLine}{0}}}
\newcommand{\AlgoNoResetCount}{\renewcommand{\@ResetCounterIfNeeded}{}}
\newcounter{AlgoSavedLineCount}

\makeatother

\newcommand{\E}{\mathsf{E}}
\providecommand{\eps}{\varepsilon}

\renewcommand{\P}{\operatorname{P}}

 \newcommand{\N}{\mathbb N}
\newcommand{\R}{\mathbb R}

\newcommand{\Id}{\mathrm{I}}

\newcommand{\CGMY}{\text{I}}
\newcommand{\In}{\mathds{1}}

\newcommand{\T}{\top}

\newcommand{\sech}{\operatorname{sech}}

\newtheorem{proposition}{Proposition}
\newtheorem{corollary}{Corollary}
\newtheorem{remark}{Remark}
\newtheorem{example}{Example}
\theoremstyle{definition}

\begin{document}

\RestyleAlgo{ruled}
\SetAlgoVlined

\begin{frontmatter}



\title{Fourier transform MCMC, heavy tailed distributions, and geometric ergodicity\tnoteref{t1}}

\author[ude,hse]{Denis Belomestny}
\author[hse]{Leonid Iosipoi}
\address[ude]{Duisburg-Essen University, Germany} 
\address[hse]{National Research University Higher School of Economics, Russia}

\tnotetext[t1]{This paper was prepared within the framework of the HSE University Basic Research Program and funded by the Russian Academic Excellence Project ``5-100".}

\begin{abstract}
Markov Chain Monte Carlo methods become increasingly popular in applied mathematics as a tool for numerical integration with respect to complex and high-dimensional distributions.  However, application of   MCMC methods to heavy tailed distributions and distributions with analytically intractable densities turns out to be rather problematic. In this paper, we propose a novel approach towards the use of MCMC algorithms for distributions with analytically known Fourier transforms and, in particular, heavy tailed distributions. The main idea of the proposed approach is to use MCMC methods in Fourier domain to sample from a density proportional to the absolute value of the underlying characteristic function. 
A subsequent application of the Parseval's formula leads to an efficient algorithm for the computation of integrals with respect to the underlying density.
We show that the resulting Markov chain in  Fourier domain may be geometrically ergodic  even in the case of heavy tailed original distributions.   
We illustrate our approach by several numerical examples including multivariate elliptically contoured stable distributions.
\end{abstract}


\end{frontmatter}

\section*{Introduction}

Nowadays Markov Chain Monte Carlo (MCMC) methods have become an important tool in machine learning and statistics, see, for instance, \cite{MR2933059, MR2238833,MR2260716}.  MCMC techniques are often applied to solve integration and optimisation problems in high dimensional spaces. The idea behind MCMC is to run a Markov chain with invariant distribution equal (approximately) to a desired distribution \(\pi\)  and then to use ergodic averages to approximate integrals with respect to \(\pi\). This approach, although computationally expensive in lower dimensions, is extremely efficient in higher dimensions. In fact, most of MCMC algorithms require  knowledge of the underlying density albeit up to a normalising  constant. However  the density of \(\pi\) might not be  available in a closed form as in applications related to stable-like multidimensional distributions,  elliptical distributions, and infinite divisible distributions. The latter class includes the  marginal distributions of L\'evy and related processes which are widely used in finance and econometrics, see e.g. \cite{tankov2003financial}, \cite{nolan2013}, and \cite{BCGMR}. 
In the above situations, it is often the case that the
Fourier transform of the target distribution is known in a closed form but the density of this distribution is intractable.  
The aim of this work is to develop a novel  MCMC methodology 
for computing integral functionals with respect to such distributions. 
Compared with existing methods, see e.g.  \cite{MR2752708},
our method avoids time-consuming numerical Fourier inversion and can be applied
effectively to high dimensional problems.  The idea of the proposed approach consists in using MCMC to sample from a distribution proportional to the absolute value of the  Fourier transform of \(\pi\) and then using the Parseval's formula to compute expectations with respect to \(\pi\). It turns out that  the resulting Markov chain can possess such nice property  as geometric ergodicity even in the case of heavy tailed  distributions \(\pi\) where the standard MCMC methods often fail to be geometrically ergodic. As a matter of fact, geometric ergodicity  plays a crucial role for concentration of ergodic averages around the corresponding expectation. 
\par
The structure of the paper is as follows. In Section~\ref{general} 
we define our framework. Section~\ref{sec:alg} contains description of the proposed methodology. In Section~\ref{sec:ergodicity} we study geometric ergodicity of the proposed MCMC algorithms. In Section~\ref{sec:examples} we apply the results from Section~\ref{sec:ergodicity} to elliptically contoured stable distributions and symmetric infinitely divisible distributions. In Section~\ref{sec:num} a thorough numerical study of MCMC algorithms in Fourier domain is presented. The paper is concluded by Section~\ref{sec:disc}.
\section{General framework}
\label{general}
Let \(g\) be a real-valued function on \(\mathbb{R}^d\) and let \(\pi\) be a bounded probability density
on \(\mathbb{R}^d\). 
By a slight abuse of notation, we will use the same letter for a distribution and its probability density, but it will cause no confusion.
Our aim is to compute the expectation of \(g\) with respect to \(\pi\), that is,
\begin{eqnarray*}
V:=\E_{\pi}[g]=\int_{\mathbb{R}^d} g(x)\pi(x)\,dx.
\end{eqnarray*}
Suppose that the density  $\pi$ is analytically unknown, and we are given its Fourier transform $\mathcal{F}[\pi](u)$ instead
\[
	\mathcal{F}[\pi](u) := \int_{\mathbb{R}^d} e^{{i} \langle u,x \rangle} \pi(x) \, dx.
\]
In this case, any numerical integration algorithm in combination with 
numerical Fourier inversion for $\mathcal{F}[g](u)$ can be applied to compute $V$.
However, this approach is extremely time-consuming even in small dimensions. 
To avoid numerical Fourier inversion, one can use the well-known Parseval's theorem.
Namely, if $g \in L^1(\mathbb{R}^d)\cap L^2(\mathbb{R}^d),$
then we can write
\begin{eqnarray*}
\label{value_parseval_simple}
V=\frac{1}{{{(2\pi)^d}}}\int_{\mathbb{R}^d}\mathcal{F}[g](-u)\mathcal{F}[\pi](u)\, du.
\end{eqnarray*}
\begin{remark}
If the tails of $g$ do not decay sufficiently fast in order to guarantee that $g \in L^1(\mathbb{R}^d)\cap L^2(\mathbb{R}^d)$, one can 
use various damping functions to overcome this problem. For example, if the function $\widetilde g(x)=g(x)/(1+x^{2p})$, for some $p\in\N$, belongs to \( L^1(\mathbb{R}^d)\cap L^2(\mathbb{R}^d),\) then we have
\begin{equation*}
\label{value_parseval1}
V=\frac{1}{{{(2\pi)^d}}}\int_{\mathbb{R}^d}\mathcal{F}[\widetilde g](-u)\left(\mathcal{F}[\pi](u)+(-1)^p\frac{d^{2p}}{du^{2p}}\mathcal{F}[\pi](u)\right)\, du,
\end{equation*}
provided that \(x^{2p}\pi(x)\in L^1(\mathbb{R}^d).\)
Another possible option to damp the growth of $g$ is to multiply it by $e^{-\langle x,R \rangle}$ for some vector $R\in \mathbb{R}^d$.
The formula for $V$ in this case reads as
\begin{equation*}
\label{value_parseval2}
V=\frac{1}{{{(2\pi)^d}}}\int_{\mathbb{R}^d}\mathcal{F}[g]({i}R-u)\mathcal{F}[\pi](u-{i}R)\, du,
\end{equation*}
provided that $\pi(x)e^{\langle x,R \rangle}\in L^1(\mathbb{R}^d)$.
\end{remark}
For the sake of simplicity we assume in the sequel that  $g\in L^1(\mathbb{R}^d)\cap L^2(\mathbb{R}^d).$
If \(|\mathcal{F}[\pi](x)|\in L^1(\mathbb{R}^d)\), then it is, up to a constant, a probability density. Thus, if $\mathcal{F}[\pi](x)$ does not vanish, one can rewrite $V$ as an expectation with respect to the density \(p(x)\propto |\mathcal{F}[\pi](x)|\),
\begin{eqnarray}
\label{imp_sample}
V=\frac{C_{p}}{(2\pi)^d} \,\mathsf{E}_{X\sim p}\left[\mathcal{F}[g](-X)\frac{\mathcal{F}[\pi](X)}{|\mathcal{F}[\pi](X)|}\right],
\end{eqnarray}
where $C_{p}$ is the normalizing constant for the density $p(x)$, that is,
\[
C_{p}:=\int_{\mathbb{R}^d} |\mathcal{F}[\pi](x)|\,dx.
\]
If $p(x)$ has a simple form and there is a direct sampling algorithm for $p(x)$, 
one can use the Monte Carlo algorithm to compute $V$ using \eqref{imp_sample}.
In more sophisticated cases,  one may use rejection sampling combined with importance sampling 
strategies, see \cite{BCGMR}. However, as the dimension $d$ increases, 
it becomes harder and harder to obtain a suitable proposal distribution.
For this reason, we need to turn to  MCMC algorithms.
The development of MCMC algorithms in the Fourier domain is the main purpose of this work.

\begin{remark}
	 The formula \eqref{imp_sample} contains the normalizing constant $C_{p}$, but
	this constant can be efficiently computed in many cases. For example, if 
	$\mathcal{F}[\pi](x)$ is positive and real, 
	then the Fourier inversion theorem yields $C_{p}= (2\pi)^d \, \pi(0).$  If the value of $\pi(0)$ is not available,
	 one can use numerical Fourier inversion.
	Furthermore, $C_{p}$ can be computed using MCMC methods, see, for example, \cite{brosse2018normalizing} and references therein. 
	Note that we can compute $C_{p}$ once and then use the formula \eqref{imp_sample} for various \(g\) without recomputing \(C_{p}.\)
\end{remark} 

\section{MCMC algorithms in the Fourier domain}\label{sec:alg}
Let us describe our generic MCMC approach in the Fourier domain. 
Let $X_0,\ldots,X_{N+n}$ be a Markov chain with the invariant distribution \(p\propto |\mathcal{F}[\pi]|\).
The samples $X_0,\ldots,X_{N}$ are discarded in order to avoid starting biases. 
Here $N$ is  chosen large enough, so that the distribution of $X_{N+1},\ldots,X_{N+n}$
is close to $p$.
We will refer to $N$ as the length of the burn-in period and $n$ as
the number of effective samples. 
According to the representation \eqref{imp_sample}, we consider a weighted average estimator 
$V_{N,n}$ for \(V\) of the form
\begin{equation}\label{eq:29032018a2}
V_{n,N}=\frac{C_{p}}{(2\pi)^d}\sum_{k=N+1}^{N+n}\omega_{N,n}(k)\,\mathcal{F}[g](-X_k)\frac{\mathcal{F}[\pi](X_k)}{|\mathcal{F}[\pi](X_k)|},
\end{equation}
where $\omega_{N,n}(k)$ are (possibly non-equal) weights such that $\sum_{k=N+1}^n  \omega_{N,n}(k) =1$.
Now let us briefly describe how to produce $X_0,\ldots,X_{N+n}$ using well-known MCMC algorithms.
We will mostly focus on the Metropolis-Hastings algorithm which 
is the most popular and simple MCMC method. Many other MCMC 
algorithms can be interpreted as special cases or extensions of this algorithm.
Nevertheless, Metropolis-Hastings-type algorithms are not exhaustive. Any MCMC algorithm can be applied in this setting
and can reach better performance than the methods listed below. Consequently, the following list in no way
limits the applicability of the generic approach of the paper.

\subsection{The Metropolis-Hastings algorithm}
The \textit{Metropolis-Hastings (MH)} algorithm 
(\citet{metropolis1953equation}, \citet{h70}) proceeds as follows. 
Let $Q(x,\cdot)$ be a transition kernel of some  Markov chain and let \(q(x,y)\) be a density of \(Q,\) that is,  $Q(x, dy) \propto q(x, y) dy$.
First we set $X_0 = x_0$ for some $x_0 \in \R^d$. 
Then, given $X_k$, we generate a proposal $Y_{k+1}$ from $Q(X_n,\cdot)$. 
The Markov chain moves towards $Y_{k+1}$ with acceptance probability 
$\alpha(X_k,Y_{k+1})$, where $\alpha(x,y)  = \min\left\{  1, \frac{p(y)q(y,x)}{p(x)q(x,y)} \right\} $, 
otherwise it remains at $X_k$. The pseudo-code is shown in Algorithm \ref{alg:MH}.
\begin{center}
\begin{algorithm}[H]
 Initialize $X_0 = x_0$\;
 \For{$k = 0 $ to $N+n$}{
  Sample $u \sim \text{Uniform}[0,1]$\;
  Sample $Y_k \sim Q(X_{k},\cdot)$\;
  \eIf{$u < \alpha(X_k,Y_{k+1})$}{
   $X_{k+1} = Y_{k+1}$\;
   }{
   $X_{k+1} = X_{k}$\;
  }
 }
 Set 
 $V_{N,n}=\displaystyle\frac{C_{p}}{(2\pi)^d}\frac{1}{n}\sum_{k=N+1}^{N+n} \mathcal{F}[g](-X_k)\frac{\mathcal{F}[\pi](X_k)}{|\mathcal{F}[\pi](X_k)|}.$
 \caption{The Metropolis-Hastings algorithm in the Fourier domain}\label{alg:MH}
\end{algorithm}
\end{center}

The MH algorithm is very simple, but it requires a careful choice of the proposal
$Q$. Many MCMC algorithms arise by considering specific choices of this distribution. 
Here are several simple instances of the MH algorithm.
\begin{description}
	\item{\textit{Metropolis-Hastings Random Walk (MHRW)}.} Here the proposal density satisfies $q(x,y) = q(y-x)$ and $q(y-x) = q(x-y)$. 
	\item{\textit{Metropolis-Hastings Independence sampler (MHIS)}.} 
	Here the proposal density satisfies $q(x,y) = q(y)$, that is, $q(x,y)$ does not depend on the previous state $x$.
\end{description}
\par
The Metropolis-Hastings algorithm produces a Markov chain 
which is reversible with respect to $p(x)$, and hence $p(x)$ is a stationary distribution for this chain, see \citet{metropolis1953equation}.
\subsection{The Metropolis-Adjusted Langevin Algorithm}
The \textit{Metropolis-Adjusted Langevin algorithm (MALA)} uses   proposals related to the discretised Langevin diffusions.
The proposal kernel
$Q_k$ depends on the step $k$ and has  the form:
\[
	Q_k(x,\cdot) = \mathcal{N}\left(x + \gamma_{k+1} \nabla \log p(x),\sqrt{2\gamma_{k+1}}\, \Id_d\right),
\]
where $(\gamma_k)_{k\ge1}$ is a nonnegative sequence of time steps and $\Id_d$ is the $d\times d$ identity matrix.
The pseudo-code of the algorithm is shown in Algorithm \ref{alg:MALA}.
\begin{center}
\begin{algorithm}[H]
 Initialize $X_0 = x_0$\;
 \For{$k = 1$ to $N+n$}{
  Sample $u \sim \text{Uniform}[0,1]$\;
  Sample $Z_{k} \sim \mathcal{N}(0,1)$\;
  Sample $Y_k = X_{k} + \gamma_{k+1} \nabla \log p(X_k) + \sqrt{2\gamma_{k+1}} Z_{k+1}$\;
  \eIf{$u < \alpha(X_k,Y_{k+1})$}{
   $X_{k+1} = Y_{k+1}$\;
   }{
   $X_{k+1} = X_{k}$\;
  }
 }
 Compute $\Gamma_{N,n} = \sum_{k=N+1}^{N+n} \gamma_k$\;
 Put  $V_{N,n}=\displaystyle\frac{C_{p}}{(2\pi)^d}\sum_{k=N+1}^{N+n} \frac{\gamma_k}{\Gamma_{N,n}} \mathcal{F}[g](-X_k)\frac{\mathcal{F}[\pi](X_k)}{|\mathcal{F}[\pi](X_k)|}.$
\caption{The Metropolis-Adjusted Langevin Algorithm in Fourier domain}\label{alg:MALA}
\end{algorithm}
\end{center}
\par

The Metropolis step in MALA makes the Markov chain 
reversible with respect to
$p(x)$, and hence $p(x)$ is a stationary distribution for the chain, see \citet{metropolis1953equation}.

\section{Geometric Ergodicity of MCMC algorithms}\label{sec:ergodicity} 
In this section, we discuss necessary and sufficient conditions needed for a Markov chain 
generated by a Metropolis-Hastings-type algorithm to be geometrically ergodic. All the results 
below will be formulated for a general target distribution $\rho$ and will be applied further to both 
$\pi$ (distribution with respect to which we want to compute the expectation) and $p$ (distribution with 
a density proportional to $\mathcal{F}[\pi]$).
\par
We say that a Markov chain is geometrically ergodic if its Markov kernel $P$ converges to a stationary distribution $\rho$ exponentially fast,
that is, there exist $r \in(0,1)$ and a function $M: \R^d\to\R$, finite for $\rho$-almost every $x\in\R^d$, such that
\[
	\bigl\| P^n(x,\cdot) - \rho(\cdot) \bigr\|_{\mathrm{TV}} \leq M(x) r^n,
	\quad x\in\R^d,
\]
where $P^n(x,\cdot)$ is the $n$-step transition law of the Markov chain, that is, $P^n(x,A) = \P (X_n\in A \vert X_0=x)$,
and $\|\cdot\|_{\mathrm{TV}}$ stands for the total variation distance.
The importance of geometric ergodicity in MCMC applications lies in the fact that
it implies central limit theorem (see, for example, \cite{ibragimovlinnik1971}, \cite{tierney1994}, \cite{jones2004}) and 
exponential concentration bounds (see, for example, \cite{dedecker2015}, \cite{wintenberger2017}, \cite{havert2019})
for the estimator $V_{n,N}$ defined in \eqref{eq:29032018a2}. 

\subsection{Metropolis-Hastings Random Walk}
Geometric ergodicity of Metropolis-Hastings Random Walk was extensively studied in 
\citet{meyntweedie1993}, \citet{tierney1994}, \citet{robertstweedie1996}, \citet{mengersetweedie1996}, \cite{jarnerhansen2000}, and \cite{jarnertweedie2003}. 
We summarize the main result in the following proposition.

\begin{proposition}[MHRW]\label{MHRW}
Suppose that the target density $\rho$ is strictly positive and continuous. Suppose further that the proposal density
$q$ is strictly positive in some region around zero (that is, there exist
$\delta>0$ and $\eps>0$ such that $q(x)\ge\eps$ for $|x|\leq\delta$) 
and satisfies $\int_{\R^d} |x| q(x) dx < \infty$. The following holds.
\begin{description}
	\item[--] (Necessary condition) 
	If the Markov chain generated by the MHRW algorithm is geometrically ergodic, 
	then there exists $s>0$ such that
	\[
		\int_{\R^d} e^{s|x|} \rho(x) \, dx < \infty.
	\]
	\item[--] (Sufficient condition) 
	Let $A(x) = \{ y\in\R^d: \ \rho(y) \ge  \rho(x) \}$ be the region of certain acceptance.
	If $\rho$ has continuous first derivatives,
	\[
	\lim_{|x|\to\infty} \Bigl\langle \frac{x}{|x|} , \nabla\log \rho(x) \Bigr\rangle= -\infty
	\quad\text{and}\quad
	\liminf_{|x|\to\infty} \int_{A(x)} q(x-y) \, dy > 0,
	\]
	then the Markov chain generated by the MHRW algorithm is geometrically ergodic. 
\end{description}
\end{proposition}
\begin{proof}
The necessary and sufficient conditions for geometric ergodicity
follow from \cite[Corollary 3.4 and Theorem 4.1 correspondingly]{jarnerhansen2000}.
\end{proof}

\subsection{Metropolis-Adjusted Langevin Algorithm}
Convergence properties of  MALA were studied in  \citet{robertstweedie1996}, \citet{ad17}, \citet{dm17}, \cite{brosse2018diffusion}.
We summarize them in the following proposition.

\begin{proposition}[MALA]
\label{prop:mala-erg}
Suppose that $\pi$ is infinitely differentiable function and $\nabla \log \rho (x)$ grows not faster than a polynomial.
\begin{description}
	\item[--] (Necessary condition) If the Markov chain generated by MALA is geometrically ergodic, 
	then 
	\[
		\lim_{|x|\to\infty} \nabla \log \rho (x)  \ne 0.
	\]
	\item[--] (Sufficient condition) 
	Assume the following.
	\begin{itemize}
		\item [(a)] The function $\log \rho(x)$ has Lipschitz continuous gradient, that is,
		there exists $L>0$ such that $| \nabla\log \rho(x) - \nabla\log \rho(y) | \leq L | x-y |$  for all $x,y\in\R^d$.
		\item [(b)] The function $-\log \rho(x)$ is strongly convex for large $u$, that is, there exist $K>0$ and $m>0$ such that for all $x\in\R^d$ with $|x|>K$ and all $v\in\R^d$,
		$\langle-\nabla^2 \log \rho(x) v , v \rangle \ge m |v|^2$.
		\item [(c)]The function $\log \rho(x)$ has uniformly bounded third derivates, that is, there exists $M>0$ such that $\sup_{x\in\R^d} |\mathrm{D}^3  \log \rho(x) | \leq M$,
		where $\mathrm{D}^3$ stands for a differential operator of the third order.
	\end{itemize}
	Then the Markov chain generated by MALA is geometrically ergodic. 
\end{description}
\end{proposition}
\begin{proof}
The necessary condition for geometric ergodicity of MALA follows from \cite[Theorem 4.3]{robertstweedie1996}. 
The sufficient conditions can be found, for example, in \cite[Section 6]{brosse2018diffusion}.
\end{proof}
\section{Examples}\label{sec:examples}
\subsection{Elliptically contoured stable distributions}
The elliptically contoured stable distribution
is a special symmetric case of the stable distribution. 
We say that $X\in\R^d$ is elliptically contoured $\alpha$-stable random vector, $0<\alpha\leq2$,
if it has characteristic function $\mathcal{F}[\pi]$ given by
\[
	\mathcal{F}[\pi](u) 
	= \exp\left( - (u^{\T}\Sigma u)^{\alpha/2} + i u^{\T}\mu \right),\quad u\in\mathbb{R}^{d},
\]
for some $d \times d$ positive semidefinite symmetric matrix $\Sigma$ and a shift vector $\mu\in\R^d$.
We note that for $\alpha = 2$ we obtain the normal distribution and
for $\alpha = 1$ we obtain the Cauchy distribution. Proposition~\ref{MHRW} implies 
the following corollary.
\par
\begin{corollary}
	Let $\pi$ have an elliptically contoured $\alpha$-stable distribution, $0<\alpha\leq2$,
	with positive definite $\Sigma$,
	and let $p \propto |\mathcal{F}[\pi]|$.
	Then the following holds.
\begin{description}
	\item[--] (Original domain) The MHRW algorithm for $\pi$ is not geometrically ergodic for any $\alpha<2$ and any proposal density $q(x)$.
	\item[--] (Fourier domain) The MHRW algorithm for $p$ is geometrically ergodic for any $1<\alpha\leq2$ provided that 
	the proposal density $q(x)$ satisfies 
	\[
		\liminf_{|x|\to\infty} \int_{ \{ y\in\R^d: \ p(y) \ge p(x) \}} q(x-y) \, dy > 0.
	\]
\end{description}
\end{corollary}
\begin{proof}
The first statement follows from the fact that for $\alpha<2$, $\pi(x)\sim|x|^{-(1+\alpha)}$ as $|x|\to\infty$, 
see \cite{nolan2018}. 
Since $\pi$ does not have exponential moments, the necessary condition from Proposition~\ref{MHRW}
does not hold, and MHRW is not geometrically ergodic. The second statement also follows from 
Proposition~\ref{MHRW}, since for any $1<\alpha\leq2$,
\[
\lim_{|x|\to\infty} \Bigl\langle \frac{x}{|x|} , \nabla\log p(x) \Bigr\rangle=-\lim_{|x|\to\infty} 
 \frac{\alpha (x^\top \Sigma x)^{\alpha/2}}{|x|}=-\infty,
\]
and the proof is complete.
\end{proof}

\subsection{Symmetric infinitely divisible distributions}
Consider a symmetric measure $\nu$ on $\mathbb{R}^{d}$
satisfying $\nu\left(\{0\}\right)=0$ and $\int_{\mathbb{R}^{d}}\left|x\right|^{2}\nu(dx)<\infty$.
We say that a distribution $\pi$ is infinitely divisible and symmetric if, 
according to the L\'evy-Khintchine representation,
its has a characteristic function $\mathcal{F}[\pi]$ given by 
\begin{align*}
\mathcal{F}[\pi](u)=\exp\biggl\{ -\frac{1}{2} u^{\top} \Sigma u +{i} \mu^{\top}u  +\int_{\mathbb{R}^{d}}\left(\cos\left( x^{\top}u \right)-1\right)\nu(dx)\biggr\} ,\quad u\in\mathbb{R}^{d},
\end{align*}
where $\Sigma$ is a symmetric positive semidefinite $d\times d$ matrix and
$\mu\in\mathbb{R}^{d}$ is a drift vector.
The triplet $\left(\Sigma,\mu,\nu\right)$ is called the L\'evy-Khintchine
triplet of $\pi$.

\begin{corollary}
	Let $\pi$ have a symmetric infinitely divisible distribution with a L\'evy-Khintchine triplet $\left(\Sigma,\mu,\nu\right)$,
	and let $p \propto |\mathcal{F}[\pi]|$.
	Then the following holds.
\begin{description}
	\item[--] (Original domain) The MHRW algorithm for $\pi$ is not geometrically ergodic for any proposal density $q(x)$ if \(\nu\) does not have exponentially decaying tails,
	that is, \(\int_{{|x|>1}} e^{s|x|}\nu(dx)=\infty\) for all $s>0$.
	\item[--] (Fourier domain)  
	Assume that the L\'evy measure $\nu(dx)$ possesses a nonnegative Lebesgue density $\nu(x)$ satisfying
	\begin{equation}
	\label{eq:cond-nu-ge1}
		s^{\alpha+d}\nu(su)\longrightarrow L(u),\quad  s \to +0, 
        \end{equation}
for some $\alpha>1$ uniformly over sets in  $\mathbb{R}^d$ not containing the origin, where the function $L(u)$ fulfils 
\begin{equation}
\label{eq:cond-nu-ge2}
	\int |y| L(y)\,dy<\infty 
	\quad\text{and}\quad 
	\inf_{|e|=1}\int_{\R^d}\left(e^{\top}y\right)\sin(e^{\top}y)\,L(y)\,dy>0.
\end{equation}
If the proposal density $q(x)$ satisfies 
	\[
		\liminf_{|x|\to\infty} \int_{ \{ y\in\R^d: \ p(y) \ge p(x) \}} q(x-y) \, dy > 0,
	\]
then 	MHRW algorithm for $p$ is geometrically ergodic.
\end{description}
\end{corollary}
\begin{proof}
The first statement follows from the fact that \(\int_{x\in\R^d} e^{s|x|}\pi(x)\,dx<\infty\) for some $s>0$ if and only if \(\int_{|x|>1} e^{s|x|}\nu(dx)<\infty\) for some $s>0$, see \cite[Theorem 25.17]{sato1999}. 
Hence Proposition~\ref{MHRW} implies that the MHRW  algorithm can not be geometrically ergodic if the invariant density \(\pi\) is infinitely divisible with the L\'evy measure \(\nu\) having only polynomial tails. 
To prove the second statement, we note that 
\[
	\Bigl\langle \frac{u}{|u|} , \nabla\log p(u) \Bigr\rangle =  |u|^{-1} \left(u^\top\Sigma u-\int_{\mathbb{R}^{d}}\left(u^{\top}x\right)\sin(u^{\top}x)\,\nu(dx)\right).
\]
Let $e_u=u/|u|$. The change of variables $y = |u|x$ implies  
\begin{align*}
	\Bigl\langle \frac{u}{|u|} , \nabla\log p(u) \Bigr\rangle& \leq -|u|^{-1}\int_{\mathbb{R}^{d}}\left(u^{\top}x\right)\sin(u^{\top}x)\,\nu(dx)\\
 	&= -|u|^{\alpha-1}\int_{\R^d}\left(e_{u}^{\top}y\right)\sin(e_{u}^{\top}y)\frac{\nu(y/|u|)}{|u|^{d+\alpha}}\,dy.
\end{align*} 
According to Proposition~\ref{MHRW}, we need to show that the limit of this expression tends to $-\infty$ as $|u|\to \infty$.
In order to take a limit, we exclude a vicinity of the origin from integration. To do this, we note that, 
according to the assumption \eqref{eq:cond-nu-ge2}, there exist a small $\eps>0$ such that
\begin{align}
\label{eq:cond-nu-ge3}
	0<\int_{\{y\in\R^d: |y|>\eps\}}\left(e^{\top}y\right)\sin(e^{\top}y)\,L(y)\,dy<\infty.
\end{align} 
for any $e\in \mathbb{R}^d$ with $|e|=1.$
Since $\left(e_{u}^{\top}y\right)\sin(e_{u}^{\top}y)\ge0$ for $|y|\leq\pi/2$, $\eps$ can be chosen such that $\eps\leq\pi/2$. Hence
\begin{align*}	
	-|u|^{\alpha-1}\int_{\R^d}\left(e_{u}^{\top}y\right)\sin(e_{u}^{\top}y)\frac{\nu(y/|u|)}{|u|^{d+\alpha}}\,dy&\leq -|u|^{\alpha-1}\int_{\{y\in\R^d: |y|>\eps\}}\left(e_{u}^{\top}y\right)\sin(e_{u}^{\top}y)\frac{\nu(y/|u|)}{|u|^{d+\alpha}}\,dy.
\end{align*} 
Due to \eqref{eq:cond-nu-ge1}, \eqref{eq:cond-nu-ge2}, and \eqref{eq:cond-nu-ge3} we have
\begin{align*}
	\lim_{|u|\to \infty}\Bigl\langle \frac{u}{|u|} , \nabla\log p(u) \Bigr\rangle& 
	 \leq - \inf_{|e|=1} \int_{\{y\in\R^d: |y|>\eps\}}\left(e^{\top}y\right)\sin(e^{\top}y)L(y)\,dy \, \lim_{|u|\to \infty}|u|^{\alpha-1} 
	 = -\infty,
\end{align*} 
and this completes the proof.
\end{proof}

\begin{example}[Stable-like processes]
Consider a \(d\)-dimensional infinitely divisible distribution with  marginal L\'evy measures of a stable-like behaviour:
\[
\nu_{j}(dx_{j})=k_{j}(x_{j})\, dx_{j}=\frac{l_{j}(|x_{j}|)}{|x_{j}|^{1+\alpha}}\, dx_{j},\quad j=1,\ldots,d,
\]
where $l_{1},\ldots,l_{d}$ are some nonnegative bounded nonincreasing functions on $[0,\infty),$  $l_{j}(0)>0$ and $\alpha\in (1,2).$  We combine these marginal L\'evy measures into a \(d\)-dimensional L\'evy density \(\nu\) via a L\'evy copula $\mathcal{C}:$
\[
\mathcal{C}(\xi_{1},\ldots,\xi_{d})=2^{2-d}\left(\sum_{j=1}^{d}\left|\xi_{j}\right|^{-\theta}\right)^{-1/\theta}\left(\eta \In_{\{\xi_{1}\cdot\ldots\cdot\xi_{d}\geq0\}}-(1-\eta)\In_{\{\xi_{1}\cdot\ldots\cdot\xi_{d}<0\}}\right), \quad \theta>0, \quad \eta\in[0,1]
\]
 as 
\[
\nu(x_{1},\ldots,x_{d})=G(\Pi_{1}(x_{1}),\ldots,\Pi_{d}(x_{d}))\, k_{1}(x_{1})\cdot\ldots\cdot k_{d}(x_{d}),
\]
where $G(\xi_{1},\ldots,\xi_{d})=\left.\partial_{1}\ldots\partial_{d}\mathcal{\, C}\right|_{\xi_{1}\ldots,\xi_{d}},$ \(\Pi_{j}(x_{j})=\nu\left(\mathbb{R},\ldots,\mathcal{I}(x_{j}),\ldots\mathbb{R}\right)\mathrm{sign}(x_{j})\)
and
\[
\mathcal{I}(x)=\begin{cases}
(x,\infty), & x\geq0,\\
(-\infty,x], & x<0.
\end{cases}
\]
Since the function $G$ is homogeneous of order $1-d$, we get 
\[
s^{\alpha+d}\nu(us)\to  L(u):=G\left(\overline{\Pi}_{1}(u_{1}),\ldots,\overline{\Pi}_{d}(u_{d})\right)
\overline{k}_{1}(u_{1})\cdot\ldots\cdot\overline{k}_{d}(u_{d}), \quad s \to 0,
\]
where
\[
\overline{k}_{j}(x_{j}):=\frac{l_{j}(0)}{|x_{j}|^{1+\alpha}},\quad\overline{\Pi}_{j}(x_{j}):=\In_{\{x_{j}\geq0\}}\int_{x_{j}}^{\infty}\overline{k}_{j}(s)\: ds+\In_{\{x_{j}<0\}}\int_{-\infty}^{x_{j}}\overline{k}_{j}(s)\: ds.
\]
As can be easily checked the conditions \eqref{eq:cond-nu-ge2} hold.
\end{example}

\begin{corollary}
	Let $\pi$ have a symmetric infinitely divisible distribution with a L\'evy-Khintchine triplet $\left(\Sigma,\mu,\nu\right)$ for some positive definite $\Sigma$,
	and let $p \propto |\mathcal{F}[\pi]|$.
	Suppose that $\Sigma$ is positive definite and 
	that $\nu(dx)$ possesses a nonnegative Lebesgue density $\nu(x)$.
	Then the following holds.
\begin{description}
	\item[--] (Original domain) Suppose that $\nu(x)>0$ and $\int_{\R^d} \nu (x)\, dx <\infty,$ then the MALA algorithm for $\pi$ is not geometrically ergodic if 
	\[
	\frac{|\nabla\nu^{\star n}(x)|}{\nu(x)}\leq C^n \zeta (x), \quad n\in \mathbb{N}\cup \{0\}, \quad x\in \mathbb{R}^d,
	\] 
	with some constant $C>0$ and $\zeta (x)\to 0$ as $|x|\to \infty.$ In particular, for all $\nu$ of the form $\nu(x)\propto 1/(1+|x|^2)^p,$ $p\geq 1,$ the MALA algorithm for $\pi$ is not geometrically ergodic.
	\item[--] (Fourier domain) The MALA algorithm for $p$ is geometrically ergodic if \(\int_{|x|\geq 1} |x|^3 \nu(x)\,dx<\infty\). 
\end{description}
\end{corollary}
\begin{proof}
The first statement follows from the fact that $\pi$ has a compound Poisson distribution, and hence
\begin{eqnarray*}
	\pi(x)=e^{-\lambda}\sum_{n=0}^{\infty}\frac{\lambda^{n}}{n!}\nu^{\star n}(x)
	\quad\text{with}\quad\lambda=\int_{\R^d} \nu (x)\, dx.
\end{eqnarray*}
Therefore, 
\begin{eqnarray*}
|\nabla \log \pi (x)| = \frac{\left|\nabla\pi(x)\right|}{\pi(x)}\leq\sum_{n=0}^{\infty}\frac{\lambda^{n}}{n!}\frac{\left|\nabla\nu^{\star n}(x)\right|}{\nu(x)}\leq e^{\lambda C} \zeta(x) \to 0,
\end{eqnarray*}
and due to Proposition~\ref{prop:mala-erg} the MALA  algorithm for $\pi$ is not exponentially ergodic.
To prove the second statement, we need to check the conditions from Proposition~\ref{prop:mala-erg}.
We have for any \(u,v\in \mathbb{R}^d,\) 
\begin{align*}
	| \nabla\log p(u) - \nabla\log p(v) |  &=   \left| \Sigma (u-v)  + \int_{\R^d} x  (\sin (u^{\top}x) - \sin(v^{\top}x))\, dx\right| \\
		&\leq \left(\sigma_{\max}(\Sigma) + \int_{\R^d} |x|^2 \nu(x) \,dx \right) |u-v|,
\end{align*}
where $\sigma_{\max}(\Sigma)$ denotes the largest singular value of $\Sigma$. Hence $\log p(u)$ has Lipschitz continuous gradient and 
the condition (a) is verified. Furthermore, 
for any \(u,v\in \mathbb{R}^d,\) 
\begin{align*}
	\langle-\nabla^2 \log p(u) v , v \rangle
	&=v^\top\Sigma v +\int_{\mathbb{R}^{d}}\left|v^{\top}x\right|^{2}\cos(u^{\top}x) \nu(x)\,dx\\
	&\geq -\left(\sigma_{\min}(\Sigma)+\int_{\mathbb{R}^{d}}\cos(u^{\top}x) \frac{\left|v^{\top}x\right|^{2}}{|v|^2} \nu(x)\, dx\right)|v|^2,
\end{align*}
where $\sigma_{\min}(\Sigma)$ denotes the smallest singular value of $\Sigma$.
By assumption, $|v^{\top}x|^{2}\nu(x)/|v|^2  \in L^1(\R^d)$ for any $v\in\R^d$. Therefore, by the Riemann-Lebesgue lemma,
$\int_{\mathbb{R}^{d}}\cos(u^{\top}x) |v^{\top}x|^{2} \nu(x) / |v|^2 \, dx \to 0$ as $|u|\to\infty$. Hence 
$-\nabla^2 \log p(x)$ is strongly convex for large $u$ and the condition (b) is verified. Finally, boundness of the third-order derivates, the condition (c), 
follows directly from the assumption. The proof is complete.
\end{proof}
\section{Numerical study}
\label{sec:num}
In what follows, we consider three numerical examples:
(1) Monte Carlo methods in Original and Fourier domains, 
(2) MCMC algorithms in Original and Fourier domains, and
(3) European Put Option Under CGMY Model.
The purpose of the following examples is to support the idea that 
moving to Fourier domain might give benefits even in the case when 
the target density $\pi$ is known in a closed form but has heavy tails.

\subsection{Monte Carlo in Original and Fourier domains}

First we compare vanilla Monte Carlo in both domains 
by estimating an expectation $V=\E_\pi[g]$ with respect to the elliptically contoured stable distribution
$\pi$ for various $\alpha$. We consider a function $g$ with its Fourier transform $\mathcal{F}[g]$ given by
\begin{align}\label{num:gfunc}
	g(x) = \prod_{i=1}^d \sech\left(\sqrt\frac{\pi}{2}x_i\right)
	\quad\text{and}\quad
	\mathcal{F}[g](u) = \left(2\pi\right)^{d/2}\prod_{i=1}^d \sech\left(\sqrt\frac{\pi}{2}x_i\right),
\end{align}
where we remind that $\sech(t) = 2/(e^{t}+e^{-t})$.
This choice stems from the fact that $\sech$ is an
eigenfunction for the Fourier Transform operator. 
Hence we will compute expectation of similar functions in 
the both domains, which will make this experiment fair.
In Original domain, we estimate $V$ with $\frac{1}{n}\sum_{i=1}^n g(X_i)$, 
where $X_1,\ldots,X_n$ is an independent sample from $\pi$. Methods to sample from 
elliptically contoured stable distribution are described in \cite{nolan2013}.
In Fourier domain, we use representation \eqref{imp_sample} and 
estimate $V$ with $\frac{1}{n}\frac{C_p}{(2\pi)^d}\sum_{i=1}^n \mathcal{F}[g](X_i)\mathcal{F}[\pi](X_i)/|\mathcal{F}[\pi](X_i)|$,
where now $X_1,\ldots,X_n$ is an independent sample from $p\propto|\mathcal{F}[\pi]|$,
which is called multivariate exponential power distribution.
The normalizing constant $C_p$ can be computed directly, $C_p=\sqrt{2/\alpha}\Gamma(d/\alpha)/(2^d \pi^{d/2}\Gamma(d/2) \det(\Sigma))$,
where $\Gamma$ stands for the Gamma function.
\par
We consider $d=5$ and $d=10$. 
We let $\mu=0$ and $\Sigma = U^\T D U$,
where $U$ is a random rotation matrix and $D$ is a diagonal matrix with numbers from $1$ to $d$ on the diagonal. 
We compute $100$ estimates based on samples  of size $n=100\,000$.
The spread of this estimates for elliptically contoured stable distribution (ECSD) and multivariate exponential power distribution (MEPD)
is given in Figure~\ref{fig:mc1} and Figure~\ref{fig:mc2}.
We see that the idea of moving to Fourier domain is reasonable ---
since samples from MEPD have lower variance, we obtain better estimates for $V$.

\begin{figure}[htb]
\includegraphics[width=0.196\linewidth]{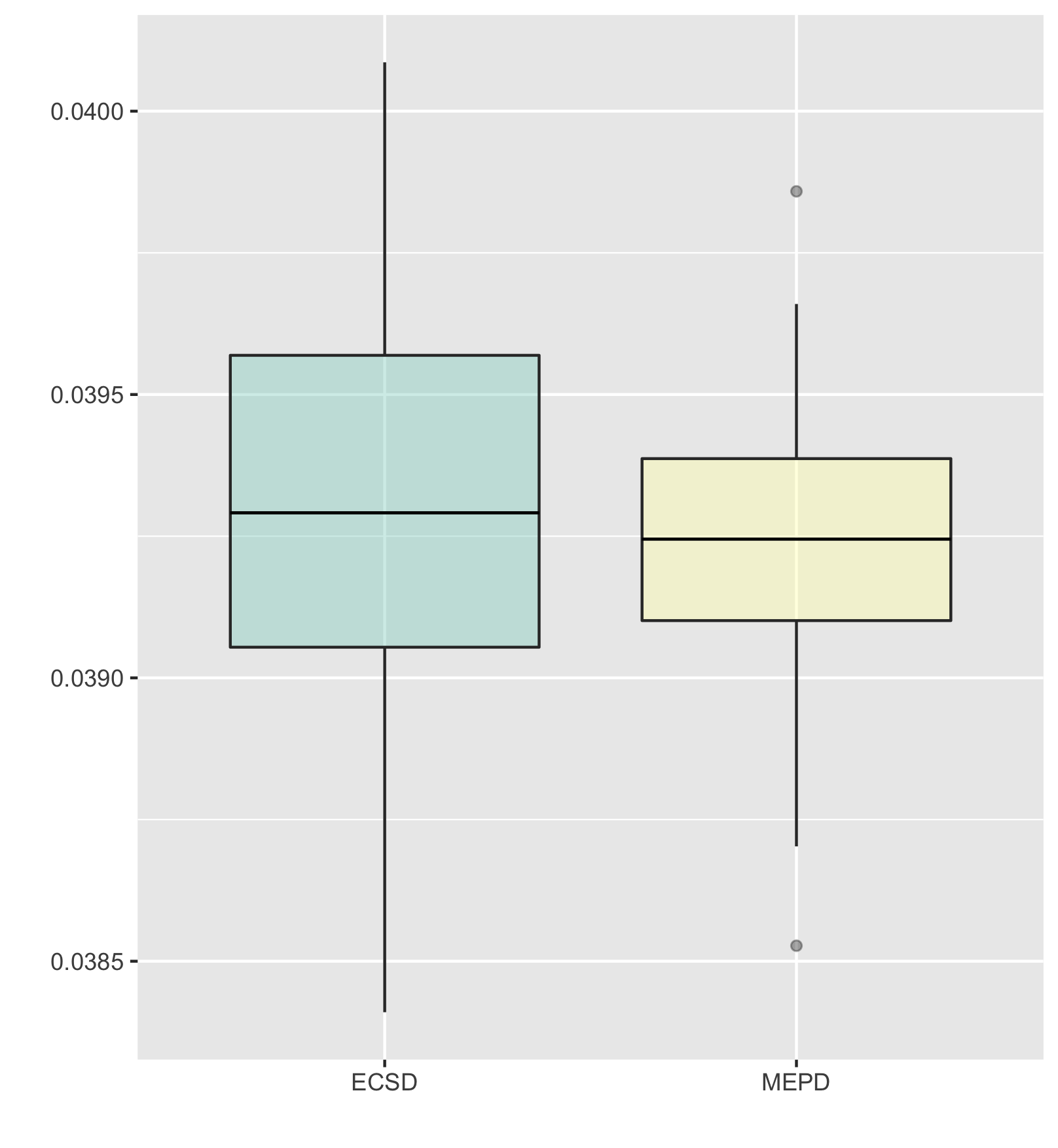}
\includegraphics[width=0.196\linewidth]{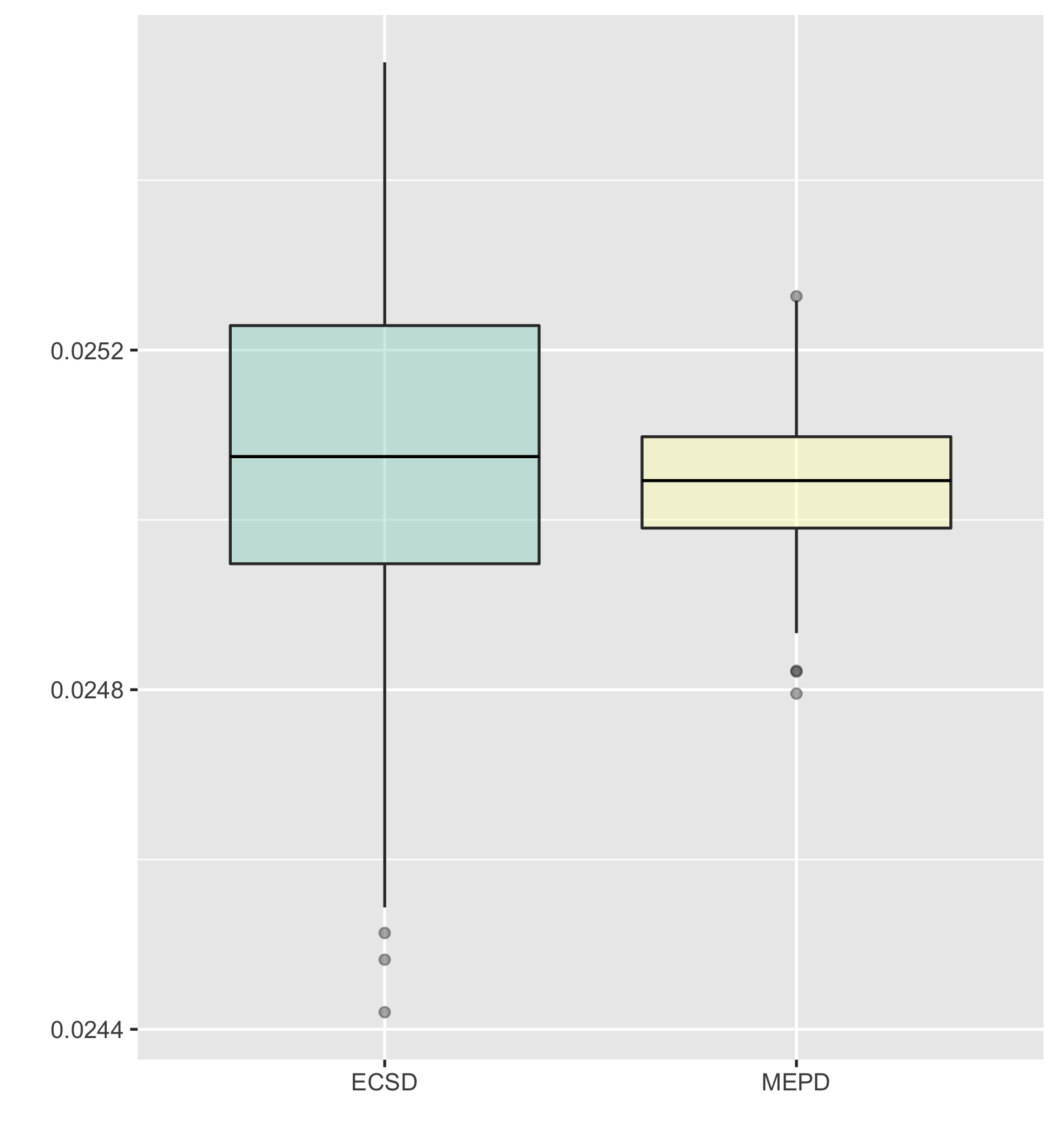}
\includegraphics[width=0.196\linewidth]{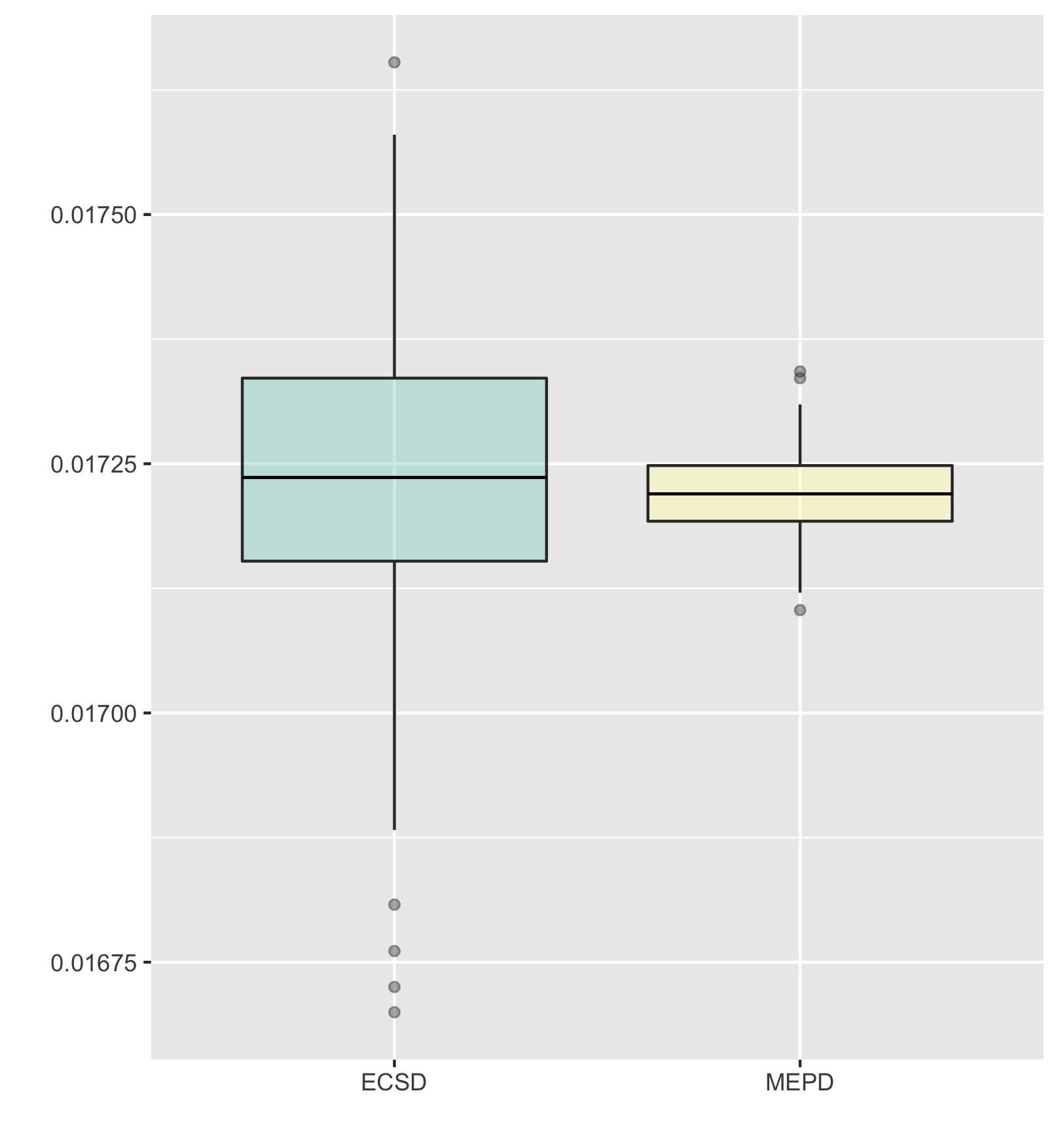}
\includegraphics[width=0.196\linewidth]{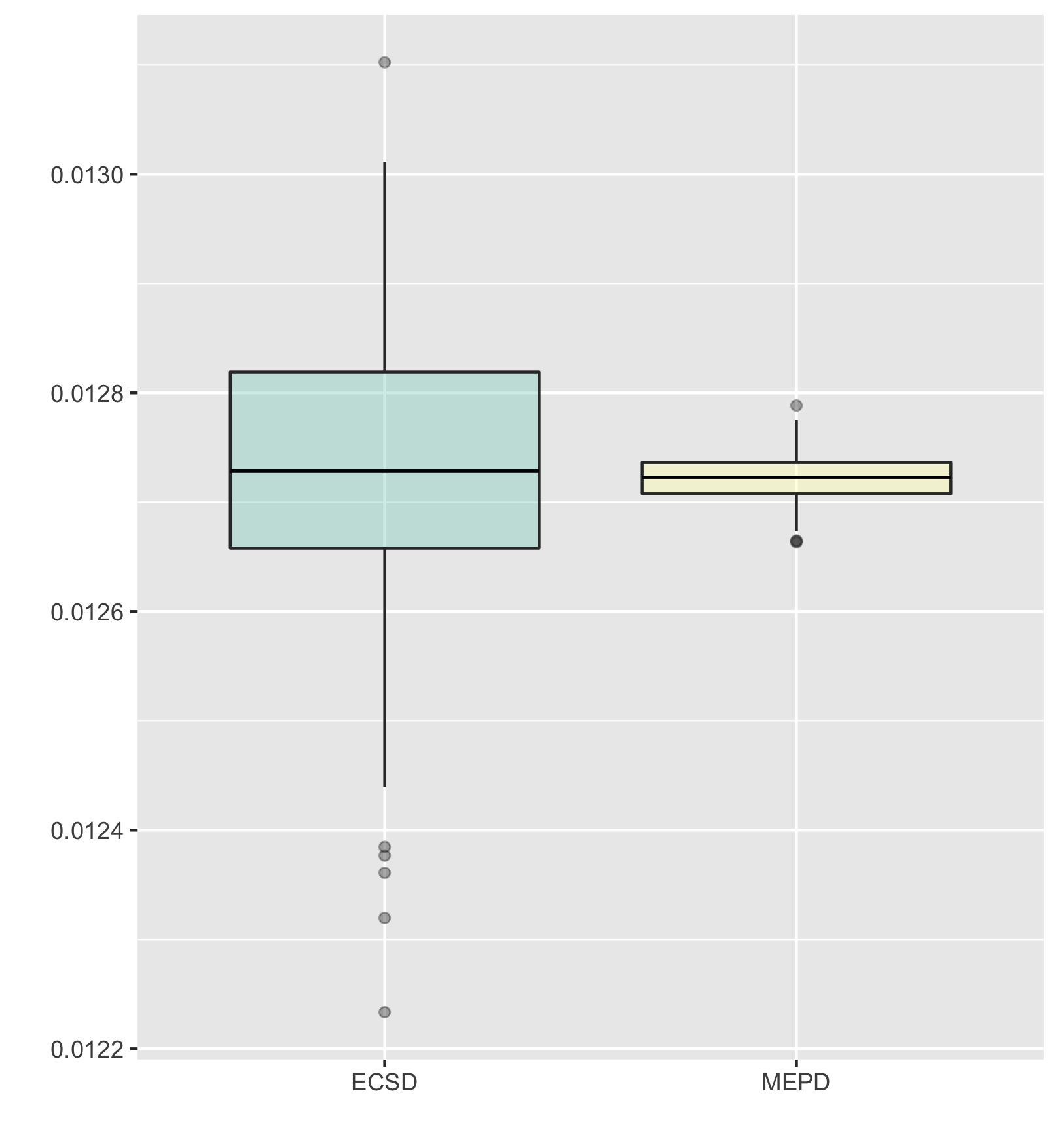}
\includegraphics[width=0.196\linewidth]{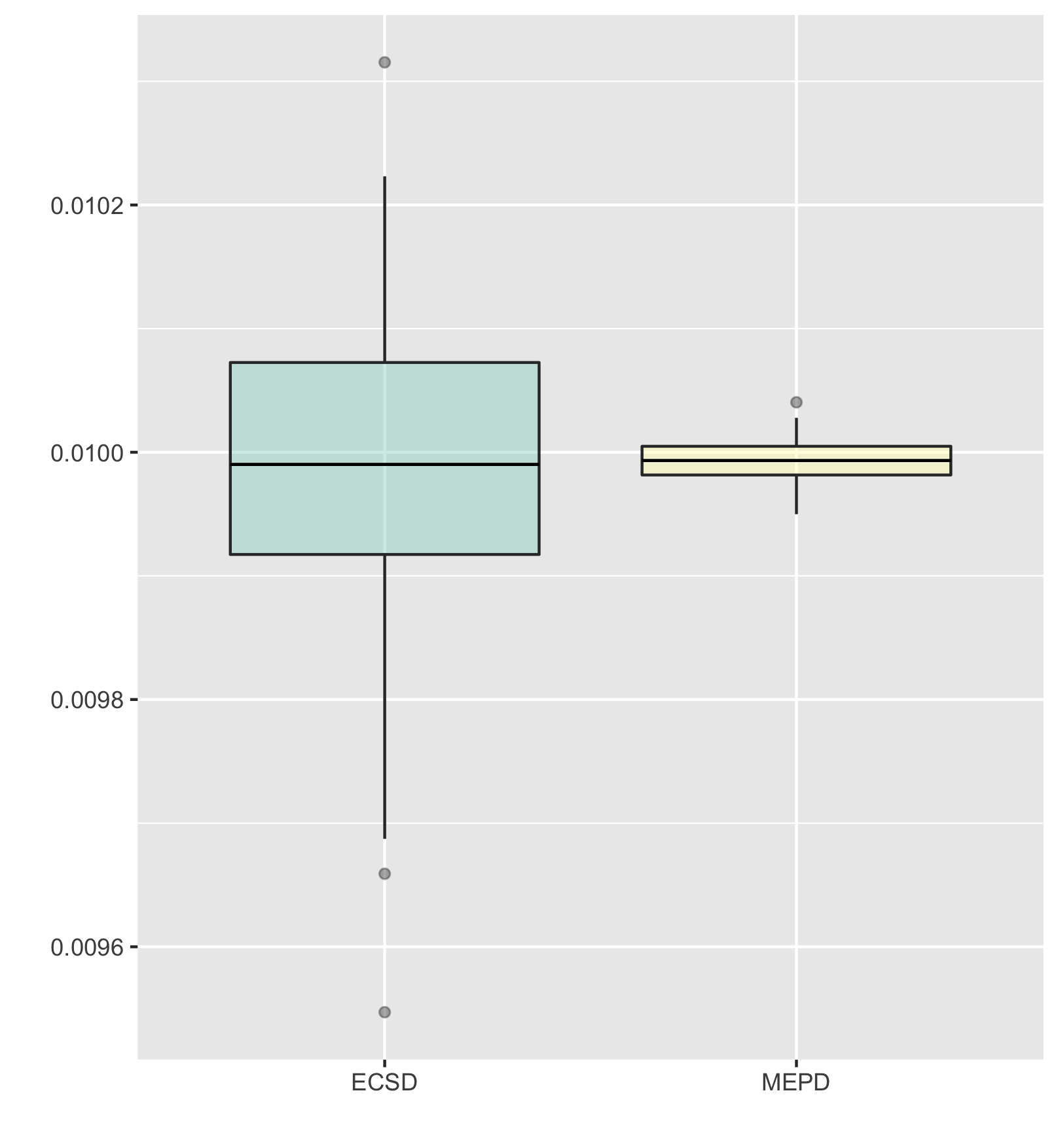}
\caption{\label{fig:mc1} Monte Carlo estimation of $\E_\pi[g]$ in $d=5$ for elliptically contoured stable distribution $\pi$ with $\alpha = 1,\, 1.2,\, 1.4,\, 1.6,\, 1.8$.}
\end{figure}

\begin{figure}[htb]
\includegraphics[width=0.196\linewidth]{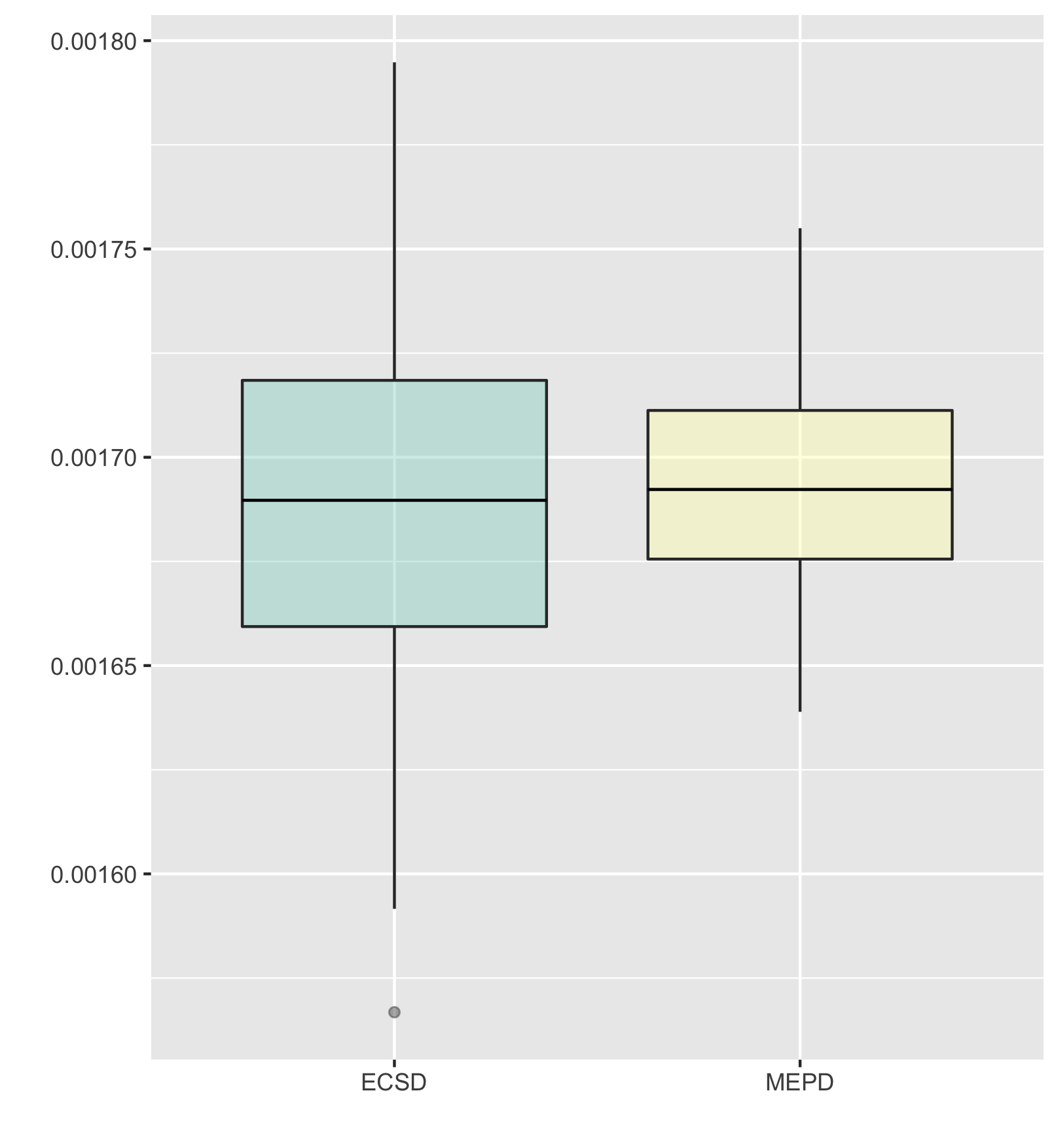}
\includegraphics[width=0.196\linewidth]{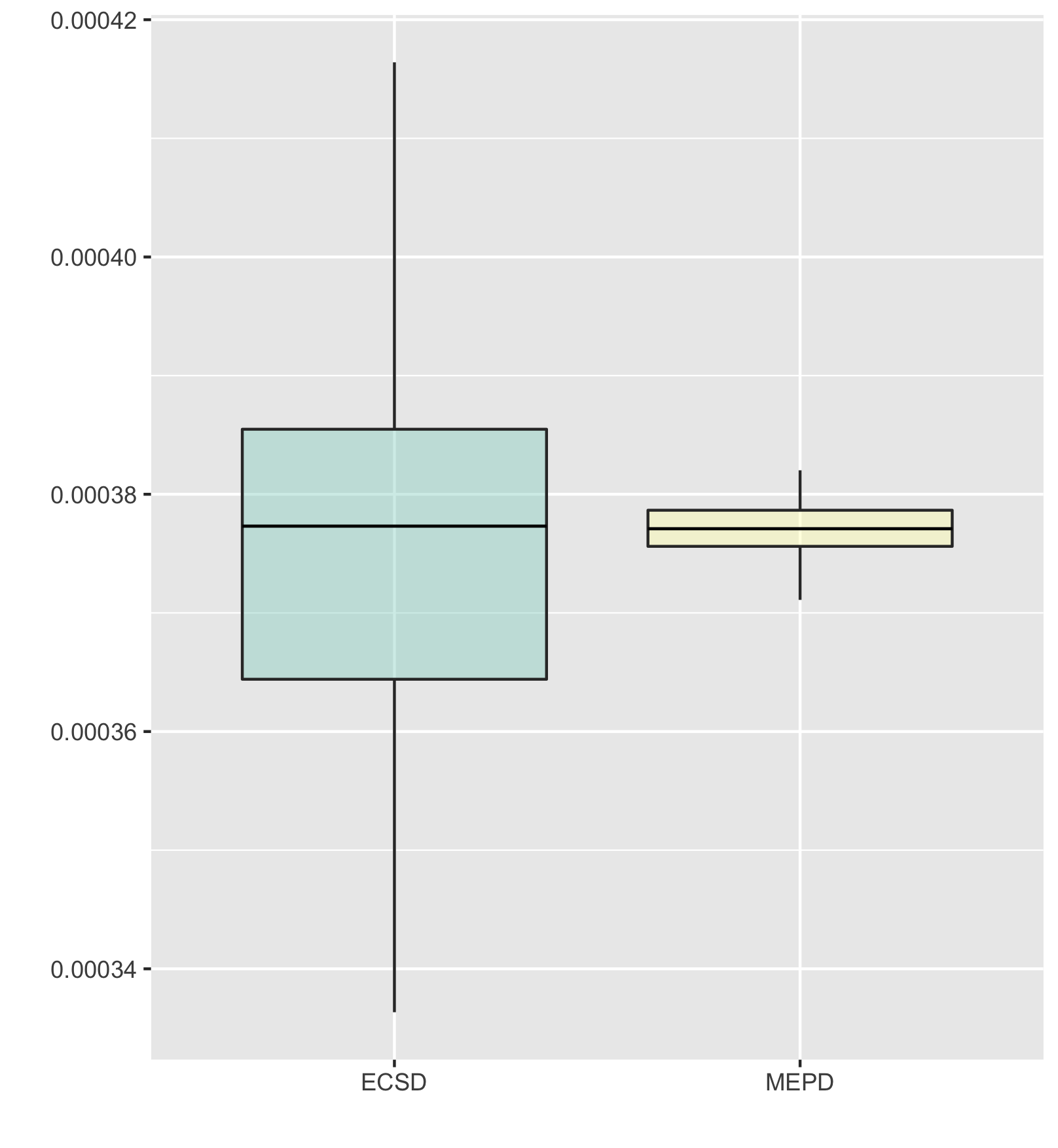}
\includegraphics[width=0.196\linewidth]{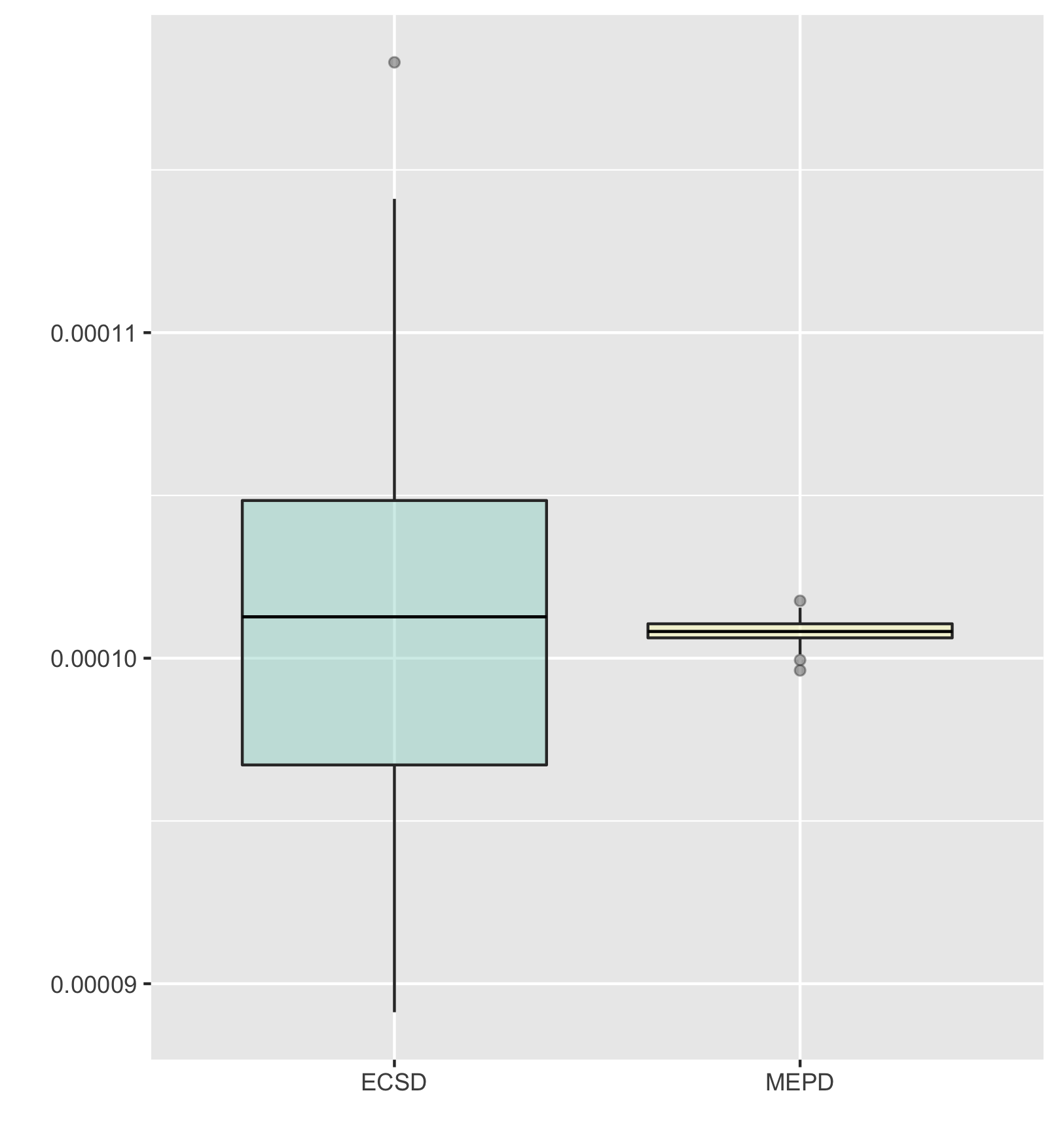}
\includegraphics[width=0.196\linewidth]{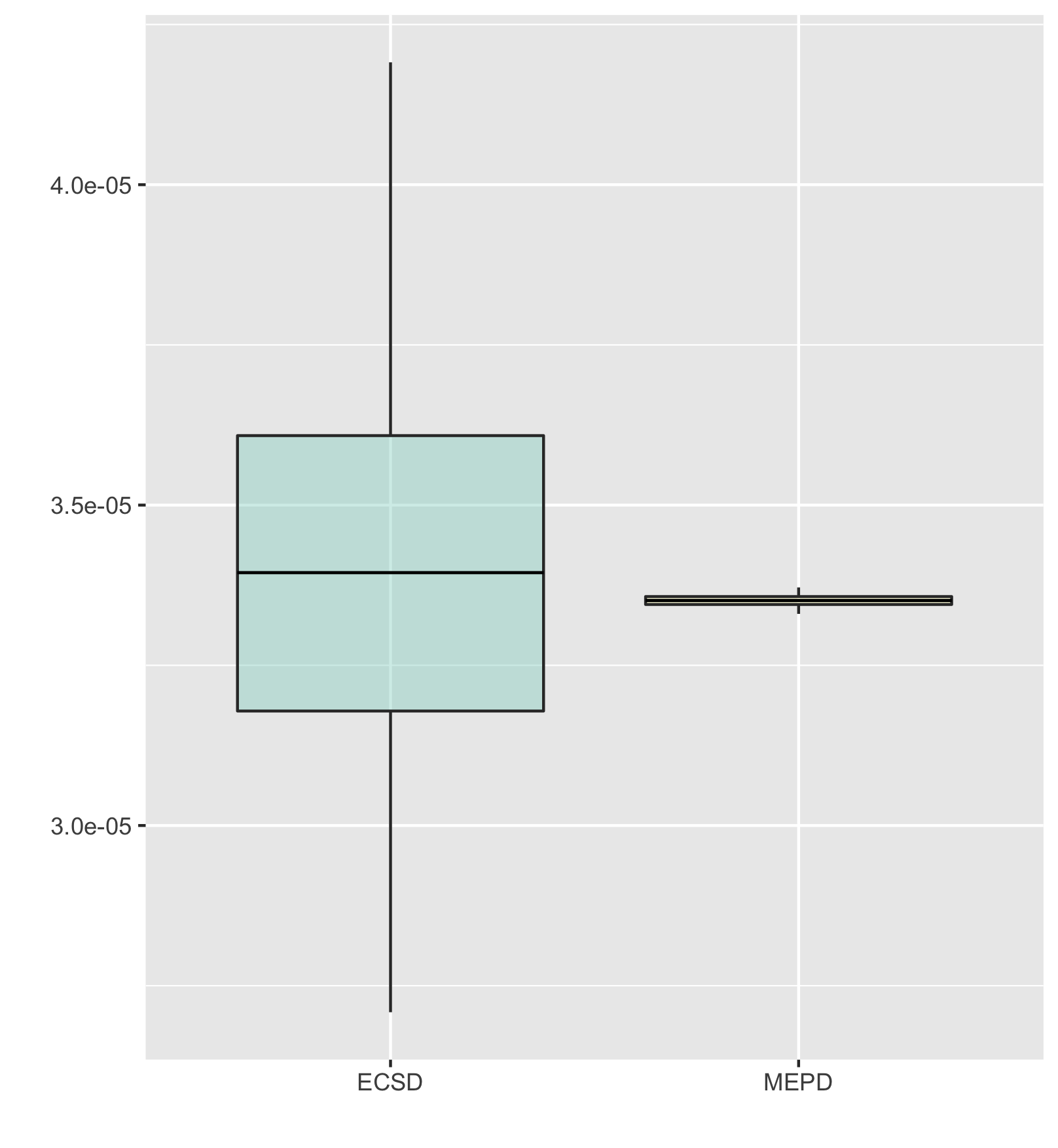}
\includegraphics[width=0.196\linewidth]{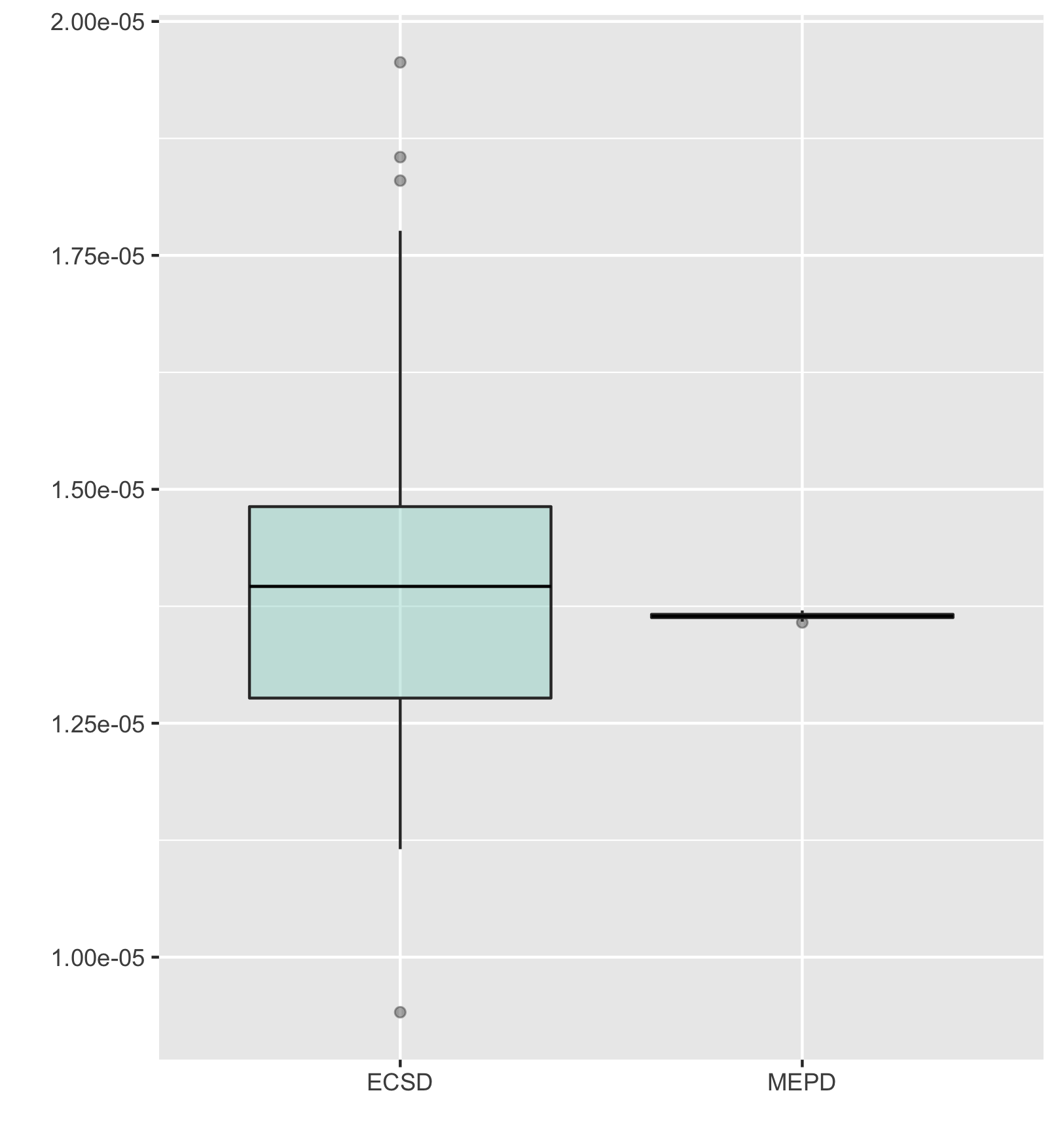}
\caption{\label{fig:mc2} Monte Carlo estimation of $\E_\pi[g]$ in $d=10$ for elliptically contoured stable distribution $\pi$ with $\alpha = 1,\, 1.2,\, 1.4,\, 1.6,\, 1.8$.}
\end{figure}

\subsection{MCMC in Original and Fourier domains}
Here we consider a specific case of elliptically contoured stable distributions with $\alpha = 1$, the Cauchy distribution.
In this case, the density $\pi$ is no longer intractable and is given by
\begin{align*}
	\pi(x) = \frac{\Gamma(\frac{d+1}{2})}{\pi^{\frac{d+1}{2}}(\det(\Sigma))^{\frac{1}{2}}} \,
		\frac{1}{(1+(x-\mu)^{\T}\Sigma^{-1}(x-\mu))^{\frac{d+1}{2}}},
\end{align*}
where, we recall, $\mu\in\R^d$ is a shift vector and $\Sigma$ is a $d \times d$ positive-definite matrix.
Its Fourier transform is given by
\begin{align*}
	\mathcal{F}[\pi](u) =  \exp\left( - (u^{\T}\Sigma u)^{1/2} + i u^{\T}\mu \right).
\end{align*}
Our aim is still an estimation of $V=\E_\pi[g]$ for the same function $g$ given in \eqref{num:gfunc},
but now we will use MCMC algorithms. 
In Original domain, we estimate $V$ with $\frac{1}{n}\sum_{i=1}^n g(X_i)$, 
where $X_1,\ldots,X_n$ is a Markov chain generated by MALA, MHRW or MHIS algorithms with the target distribution $\pi$. 
In Fourier domain, we 
estimate $V$ with $\frac{1}{n}\frac{C_p}{(2\pi)^d}\sum_{i=1}^n \mathcal{F}[g](X_i)\mathcal{F}[\pi](X_i)/|\mathcal{F}[\pi](X_i)|$, see \eqref{imp_sample},
where now $X_1,\ldots,X_n$ is a Markov chain generated by the same MCMC algorithms with the target distribution $p\propto|\mathcal{F}[\pi]|$.
\par
The experiment is organized as follows. 
We chose $\mu=0$ and $\Sigma = U^\T D U$,
where $U$ is a random rotation matrix and $D$ is a diagonal matrix with numbers from $0.2$ to $0.2d$ on the diagonal. 
The matrix $\Sigma$ is chosen so to prevent large values for $C_{p}$.
We start with computing of a gold estimate for $V$.
This is done by averaging $100$ vanilla Monte Carlo estimates of size $n=100\,000$.
For MCMC algorithms, we generate $100$ independent trajectories of size $n=105\,000$,
where the first $n=5\,000$ steps are discarded as a burn-in. For both MHRW and MHIS we use the normal proposal.
Parameters for MCMC algorithms are chosen adaptively by minimizing the MSE between the gold estimate and
$100$ estimates computed for each trajectory. The resulting parameters can be viewed as a best possible
parameters, and the performance of an MCMC algorithm as a best possible. Once parameters are estimated, 
we generate $100$ new independent trajectories of the same size and compute estimates of $V$ for 
each of them. The relative error of these estimates (with respect to the gold estimate) for $d=5,\,10,\,15$ is shown on boxplots in Fugure~\ref{fig:mcmc1}. 
We see that MCMC algorithms do not work when the target density 
has heavy tails, hence moving to Fourier domain might be the only possible option (if, for example, there is no direct sampling algorithm from $\pi$).
\begin{figure}[htb]
\includegraphics[width=0.33\linewidth]{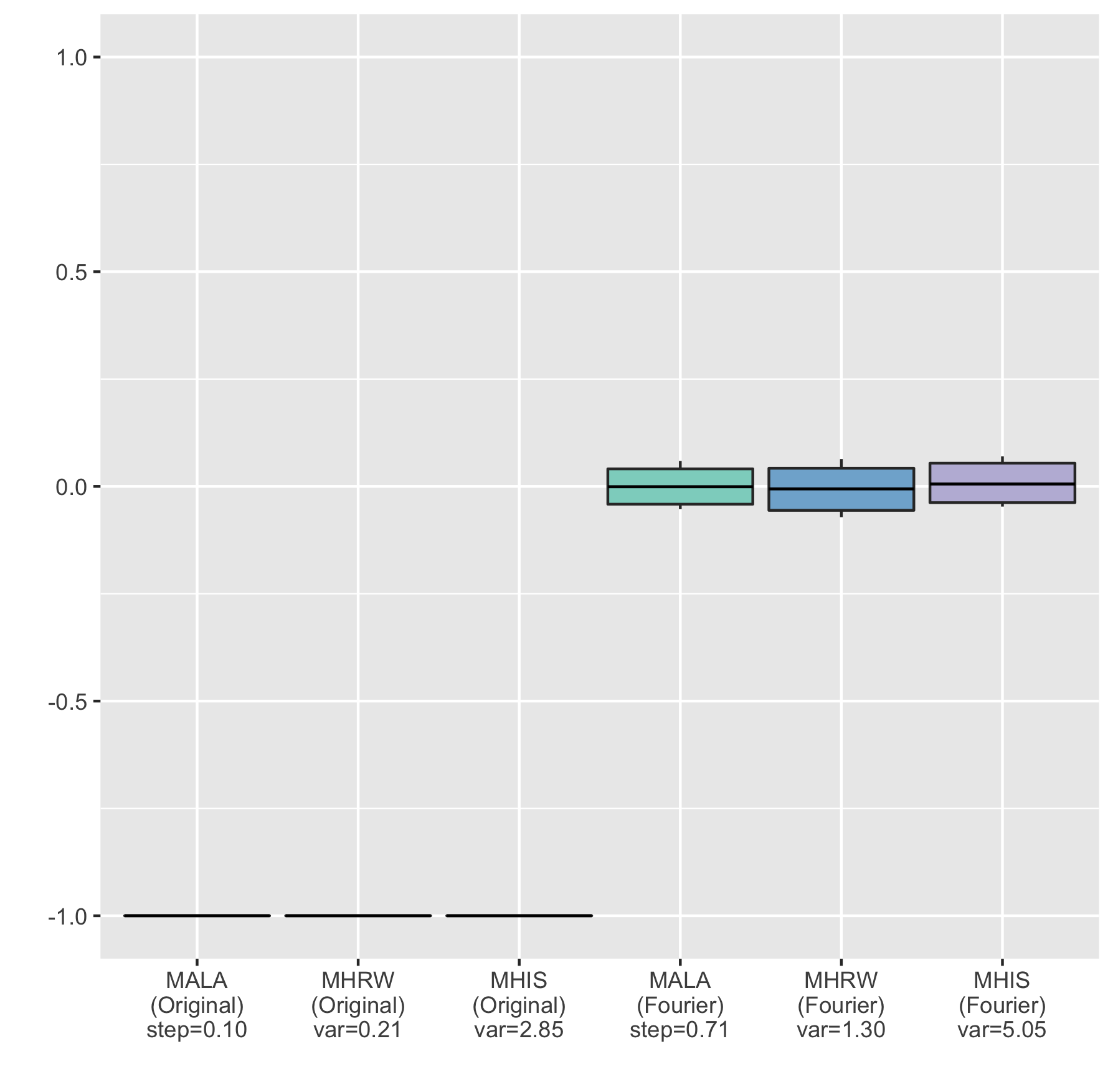}
\includegraphics[width=0.33\linewidth]{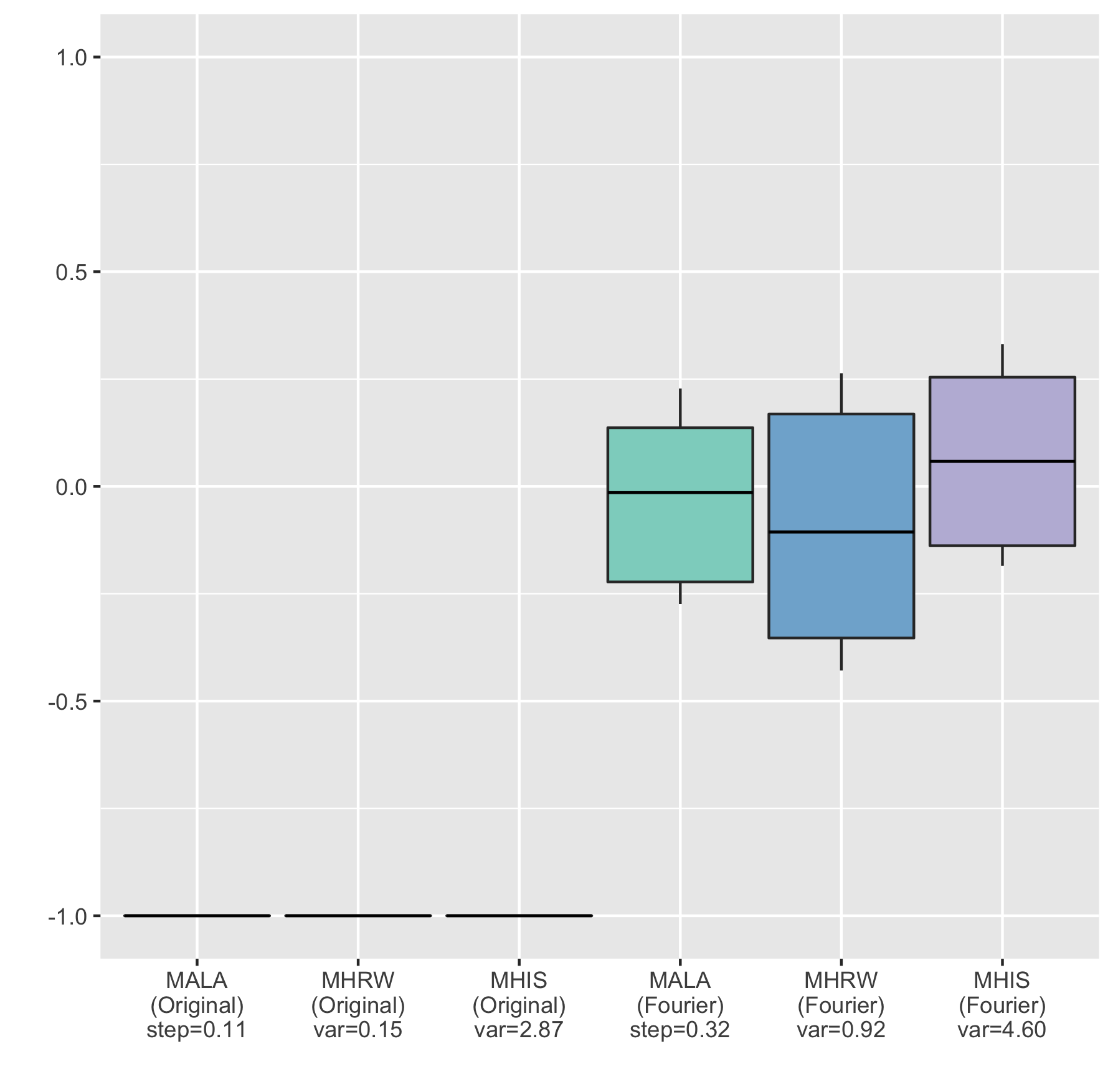}
\includegraphics[width=0.33\linewidth]{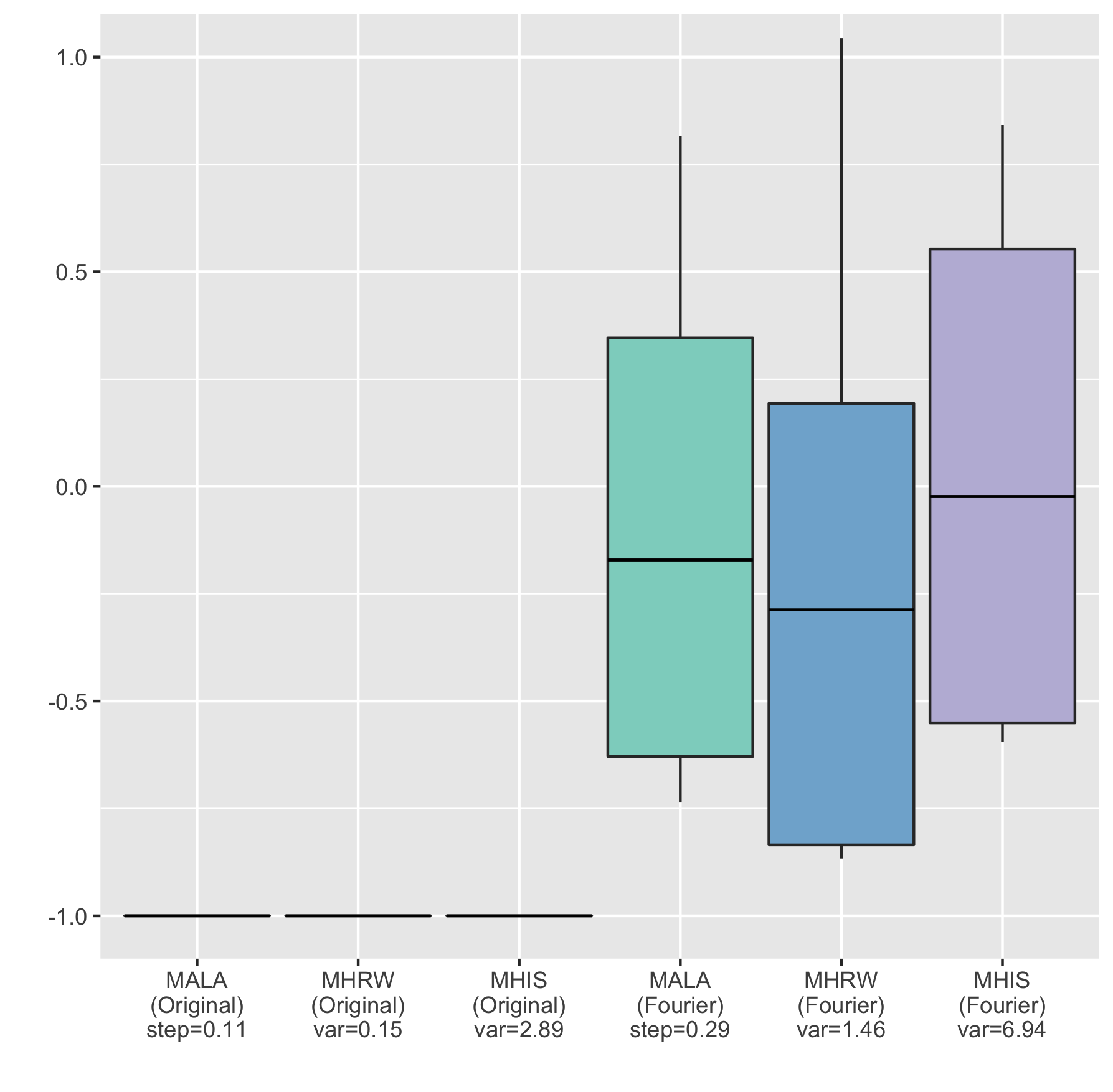}
\caption{\label{fig:mcmc1} MCMC estimation of $\E_\pi[g]$ in $d=5,\,10,\,15$ for the Cauchy distribution.}
\end{figure}

\subsection{European Put Option Under CGMY Model}
Now we consider a Financial example from  \cite{belomestny16}.
The CGMY process $\{X_t\}_{r\ge0}$ is a pure jump process with the L\'evy measure
\[
	\nu_{\CGMY}(x) = C \biggl[ \frac{e^{Gx}}{|x|^{1+Y}}\In_{x<0}+ \frac{e^{-Mx}}{|x|^{1+Y}}\In_{x>0} \biggl],
\]
where $C,G,M > 0$, $0 < Y < 2$, see \cite{cgmy2002} for details on CGMY processes.
The characteristic function of $X_T$ reads as
\begin{align*}
	\mathcal{F}[\pi](u) = 
	\exp\bigl( {i}\mu u T + TC\Gamma(-Y)
	[(M - iu)^Y - M^Y + (G + iu)^Y - G^Y] \bigr),
\end{align*}
where the drift $\mu\in\R^d$ is given for some $r>0$ by
\[
	\mu = r - C\Gamma(-Y)[(M - 1)^Y - M^Y + (G + 1)^Y - G^Y] .
\]
Suppose that the stock prices follow the model
\[
	S^k_t = e^{X^k_t}, \quad k=1,\ldots,d,
\]
where $X^k_t$ are independent CGMY processes.
Let $g(x)$ be the payoff function for the put option on the maximum of $d$ assets, i.e.,
\[
g(x) = (K  - e^{x_1} \vee\ldots\vee e^{x_d})^{+}.
\]
Our goal is to compute the price of the European put option which is given by
$
	V = e^{-rT} \E\bigl[g(X^1_T,\ldots,X^d_T)\bigr].
$
Application of the Parseval's formula with damping the growth of $g(x)$ by $e^{\langle x,R \rangle}$ for some vector $R\in \mathbb{R}^d$
leads to the formula
\begin{equation*}
V=\frac{e^{-rT}}{(2\pi)^d}\int_{\mathbb{R}^d}\mathcal{F}[g]({i}R-u)\mathcal{F}[\pi](u-{i}R)\, du.
\end{equation*}
To ensure the finiteness of $\mathcal{F}[\pi](u-{i}R)$, we need to choose $R$ such that 
its coordinates satisfy $-G<R_k<0$, $k=1,\ldots,d$.
The authors in \cite{belomestny16} propose to use importance sampling strategy with the following representation 
\[
	V = \frac{e^{-rT}}{(2\pi)^d} \,\mathsf{E}_{X\sim q}\left[\mathcal{F}[g]({i} R-X)\frac{\mathcal{F}[\pi](X-{i}R)}{|q(X)|}\right],
	\quad\text{where}\quad
	q(u) = \frac{1}{2\theta^{1/\alpha}\Gamma(1+1/\alpha)} e^{-\frac{|u|^\alpha}{\theta}}
\]
for some parameters $\alpha, \theta>0$.
Our goal is to compare this approach with the Fourier transform MCMC strategy
\begin{equation*}
	V = \frac{C_p e^{-rT}}{(2\pi)^d} \,\mathsf{E}_{X\sim p}\left[\mathcal{F}[g]({i} R-X)\frac{\mathcal{F}[\pi](X-{i}R)}{|\mathcal{F}[\pi](X-{i}R)|}\right],
\end{equation*}
where \(p(u)\propto |\mathcal{F}[\pi]({i}R-u)|\). 
The Fourier transform of the payoff function $g(x)$  
is given by
\begin{align*}
	\mathcal{F}[g]({i}R-u) &= \frac{(-1)^{d+1}K^{1-\sum_{k=1}^d (R_k-{i}u_k)}}{(\sum_{k=1}^d (R_k-{i}u_k) - 1)\prod_{k=1}^d(R_k+{i}u_k)},
\end{align*}
see \cite[Appendix A]{egp10} for the proof.
\par
We take $C = 1$, $G = 5$, $M = 5$, $Y = 0.5$, $r = 0.1$, $S_0 =100$, $K =100$, 
$T = 1$, $R=-(1.5,\ldots,1.5)$, and compare importance sampling (IS) with Metropolis-Hastings Random Walk (MHRW)
by generating $100$ independent trajectories of size $n = 10\,000$ and computing estimates of $V$ for 
each of them. We use the normal density for both importance sampling density 
(this corresponds to $\alpha=2$, $\theta=2$) and proposal density in MHRW.
The burn-in period for MHRW is $N=5\,000$. 
The spread of the estimates is presented in Figure \ref{fig:is1}.
\begin{figure}[htb]
\includegraphics[width=0.245\linewidth]{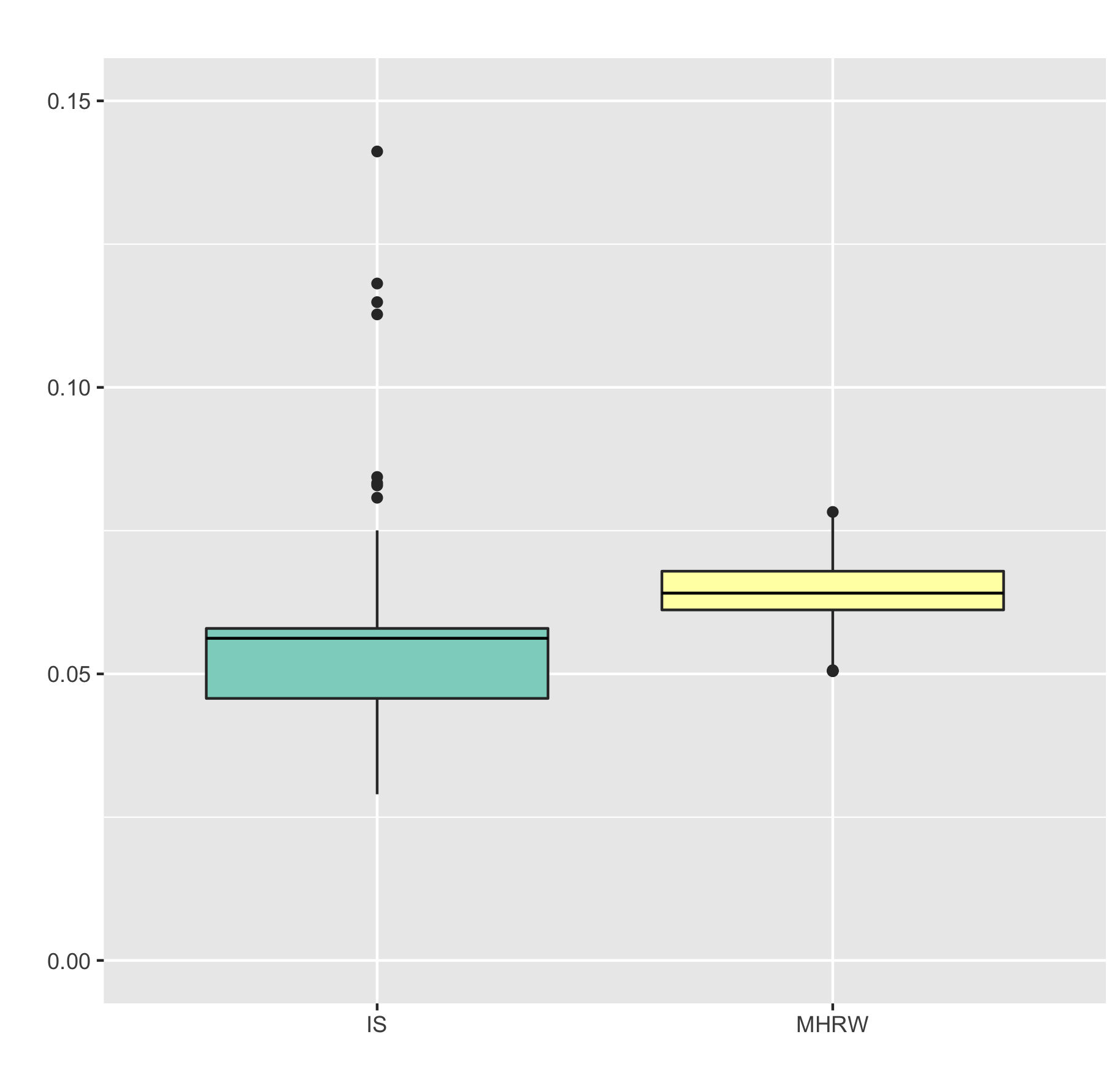}
\includegraphics[width=0.245\linewidth]{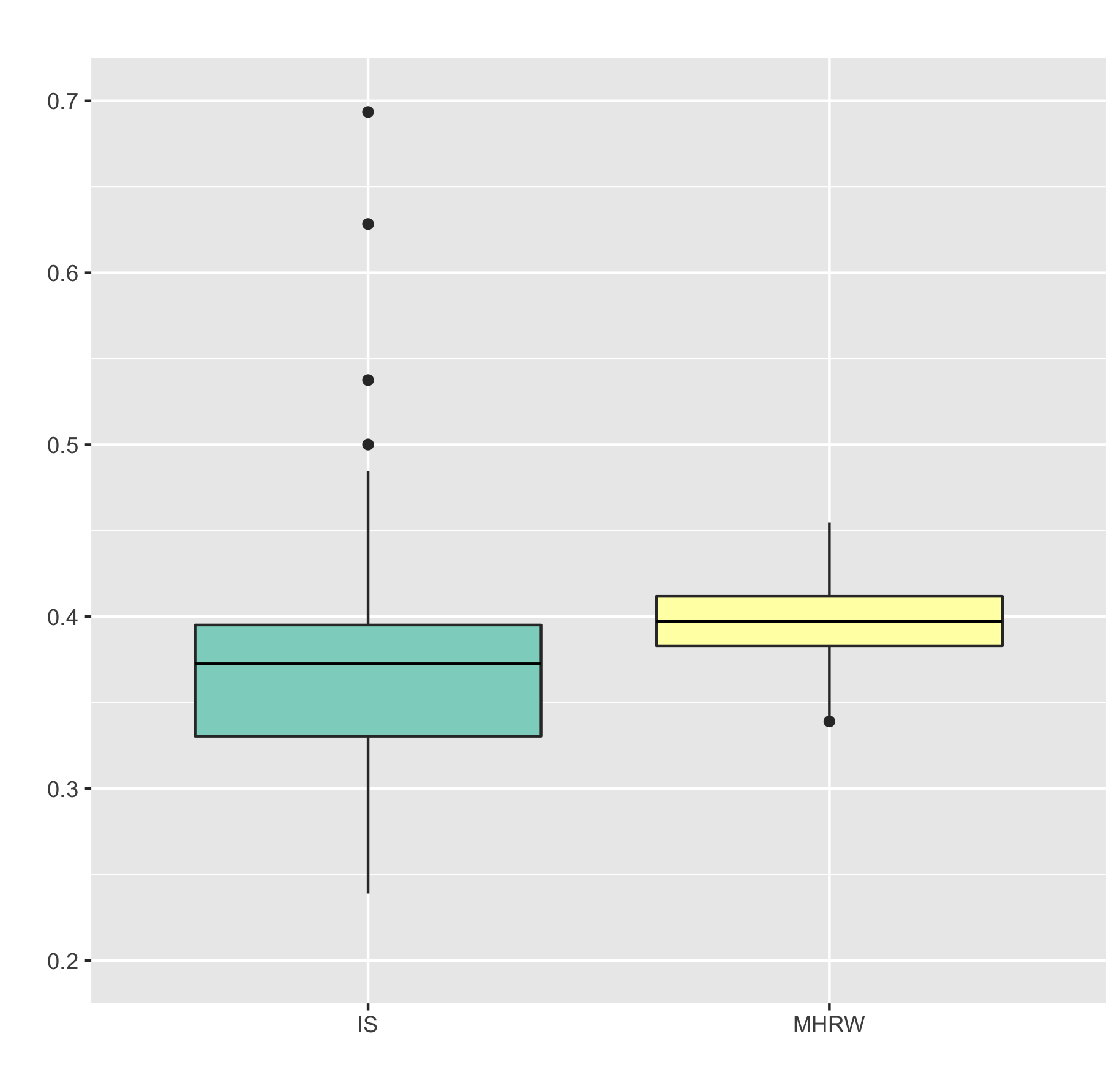}
\includegraphics[width=0.245\linewidth]{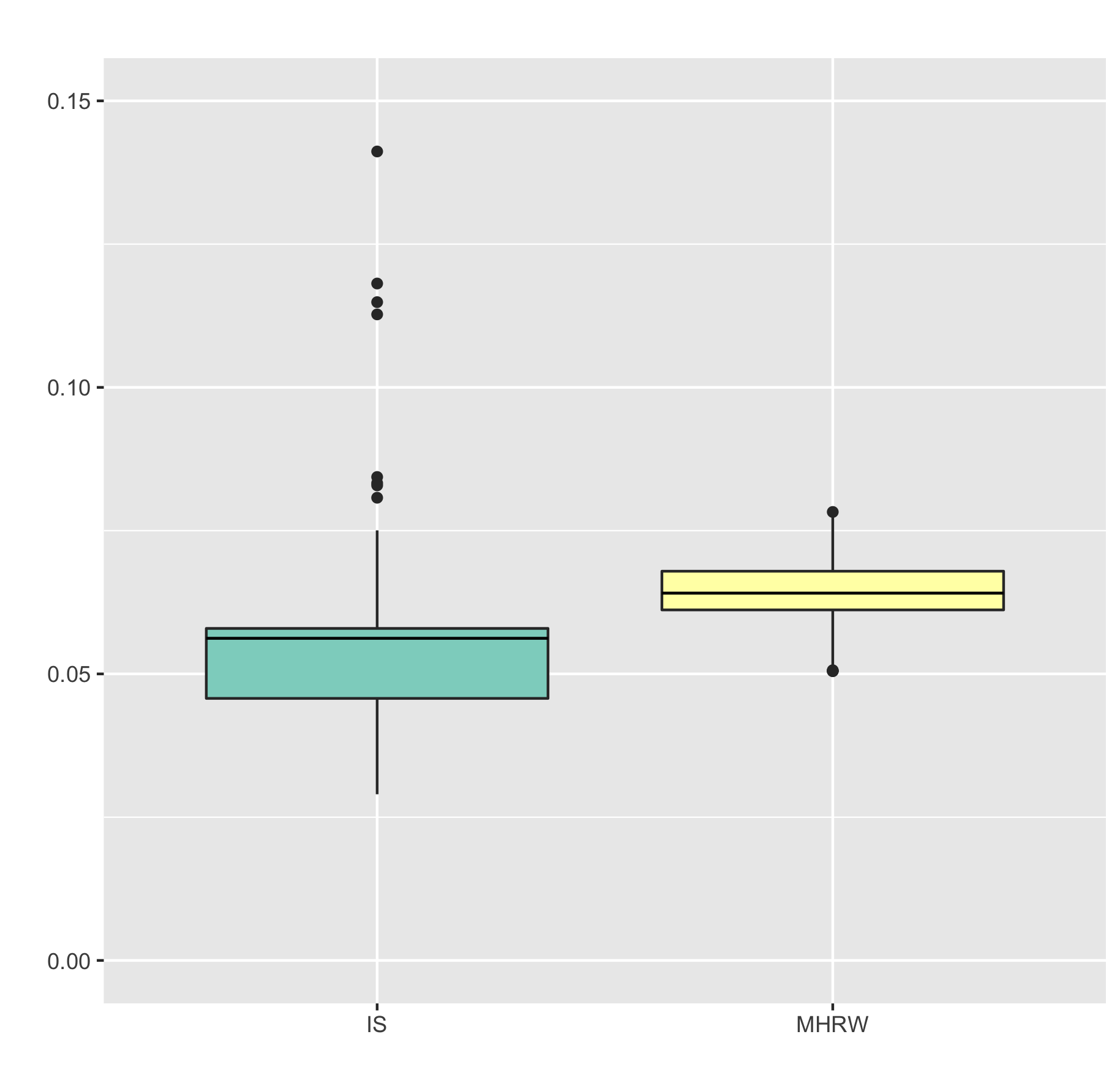}
\includegraphics[width=0.245\linewidth]{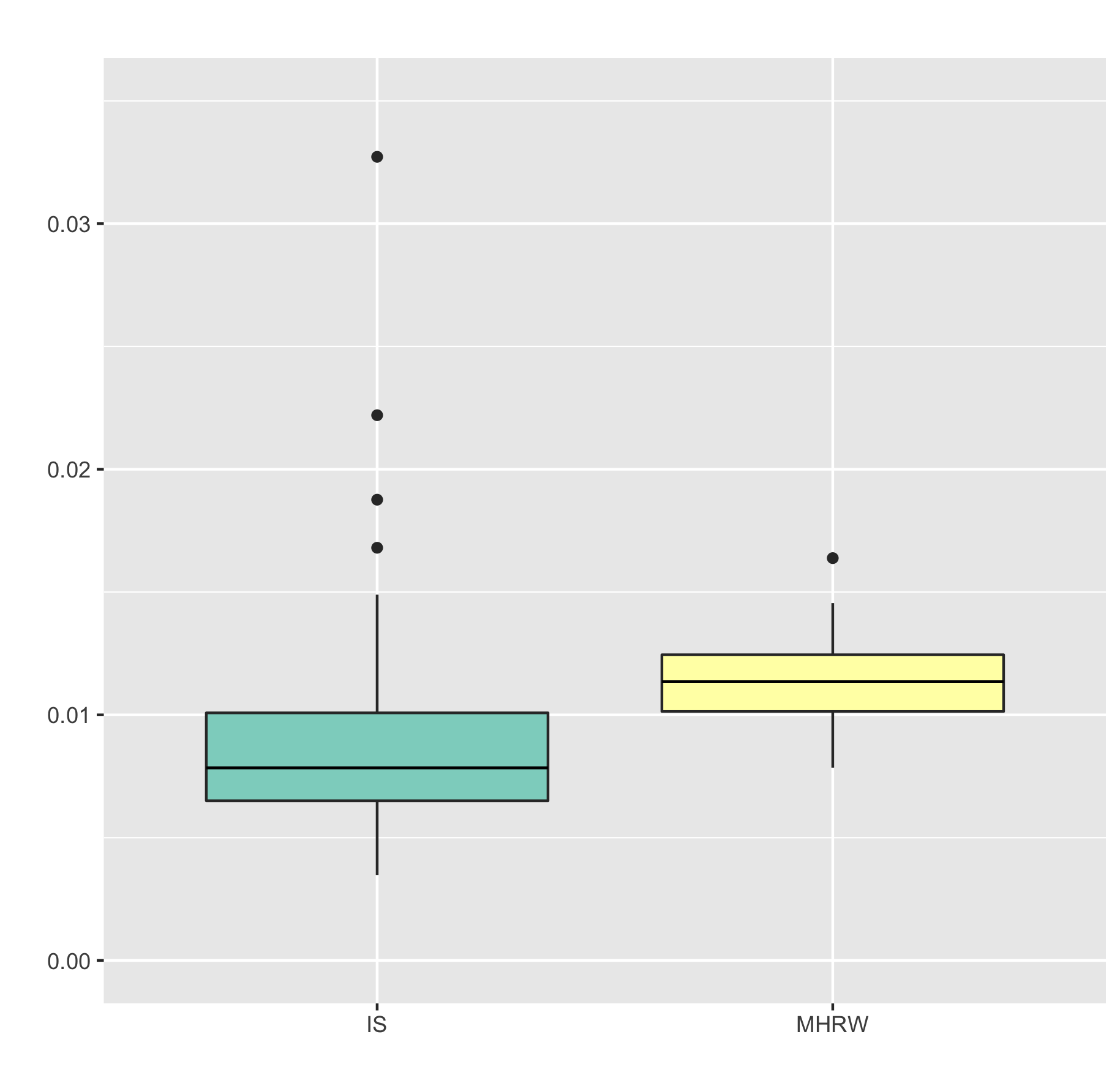}
\caption{\label{fig:is1} Put option on the maximum of $d$ assets in CGMY model for $d=2, 4, 6$, and $8$.}
\end{figure}

\section{Discussion}
\label{sec:disc}
We proposed a novel MCMC  methodology for the computation of expectation with respect to distributions with analytically known Fourier transforms. The proposed approach is rather general and can be also used in combination with importance sampling as a variance reduction method. As compared to the  MC method in spectral domain, our approach requires only generation of simple random variables and therefore is computationally more efficient.  Finally let us note that our methodology may also be useful in the case of heavy tailed distributions with analytically known Fourier transforms.

\section*{References}
\bibliographystyle{elsarticle-harv} 
\bibliography{Fourier_MCMC}
%
%
\end{document}